\documentclass[journal,12pt,onecolumn,draftclsnofoot,]{IEEEtran} %Hanaa
\normalsize

\ifCLASSINFOpdf

\else

\fi

\usepackage{array}
\usepackage{color}
\usepackage{amsfonts}
\usepackage{amssymb}
\usepackage{amsmath}
\usepackage[dvips]{graphicx}
\usepackage{cite}
\usepackage{balance}
\usepackage{caption}
\usepackage{subcaption}
\usepackage{textcomp}
\usepackage{gensymb}
\usepackage{float}
\usepackage{tabularx}
\usepackage{multirow}
\usepackage{amsthm}
\newtheorem{proposition}{Proposition}

\begin{document}
%\markboth{Submitted to \textit{IEEE Wireless Communications Letters}}{Karas, Karagiannidis, Smart Decode-and-Forward Relaying with Polar Codes}
\title{Performance Analysis of Coherent and Noncoherent Modulation under I/Q Imbalance}

\author{Bassant Selim, \IEEEmembership{Student Member, IEEE},  Sami Muhaidat, \IEEEmembership{Senior Member, IEEE}, Paschalis C. Sofotasios, \IEEEmembership{Senior Member, IEEE}, Bayan S. Sharif, \IEEEmembership{Senior Member, IEEE}, Thanos Stouraitis, \IEEEmembership{Fellow, IEEE}, George K. Karagiannidis, \IEEEmembership{Fellow, IEEE} \\ and Naofal Al-Dhahir, \IEEEmembership{Fellow, IEEE} 

\thanks{B. Selim, B. Sharif and T. Stouraitis are with the Department of Electrical and Computer Engineering, Khalifa University of Science and Technology, PO Box 127788, Abu Dhabi, UAE (e-mail: \{bassant.selim; bayan.sharif; thanos.stouraitis\}@kustar.ac.ae).}

\thanks{S.~Muhaidat is with the Department of Electrical and Computer Engineering, Khalifa University of Science and Technology, PO Box 127788, Abu Dhabi, UAE  and with the Institute for Communication Systems, University of Surrey, GU2 7XH, Guildford, UK (email: $ \rm muhaidat@ieee.org$).}

\thanks{P. C. Sofotasios is with the Department of Electrical and Computer Engineering, Khalifa
University of Science and Technology, PO Box 127788, Abu Dhabi, UAE, and  with the Department of Electronics and Communications Engineering, Tampere University of Technology, 33101 Tampere, Finland (e-mail: $\rm p.sofotasios@ieee.org$).}%

\thanks{G. K. Karagiannidis is with the Department of Electrical and Computer Engineering, Aristotle University of Thessaloniki, 54124 Thessaloniki, Greece  (e-mail: geokarag@auth.gr).}

\thanks{N. Al-Dhahir is with the Department of Electrical Engineering, University of Texas at Dallas, TX 75080 Dallas, USA (e-mail: {aldhahir@utdallas.edu}).}
}

\maketitle	

\begin{abstract}
In-phase/quadrature-phase Imbalance (IQI) is considered a major performance-limiting impairment in direct-conversion transceivers. Its effects become even more pronounced at higher carrier frequencies such as the millimeter-wave frequency bands being considered for 5G systems. In this paper, we quantify the effects of IQI on the performance of different modulation schemes under multipath fading channels. This is realized by developing a general framework for the symbol error rate (SER) analysis of coherent phase shift keying, noncoherent differential phase shift keying and noncoherent frequency shift keying under IQI effects. In this context, the moment generating function of the signal-to-interference-plus-noise-ratio is first derived for both single-carrier and multi-carrier systems suffering from transmitter (TX) IQI only, receiver (RX) IQI only and joint TX/RX IQI. Capitalizing on this, we derive analytic expressions for the SER of the different modulation schemes. These expressions are corroborated by comparisons with corresponding results from computer simulations and they provide insights into the dependence of IQI on the system parameters. We demonstrate that the effects of IQI differ considerably depending on the considered system as some cases of single-carrier transmission appear robust to IQI, whereas multi-carrier systems experiencing IQI at the RX require compensation in order to achieve a reliable communication link.
\end{abstract}


\section{Introduction}\label{S:Intro}

The emergence of the Internet of Things (IoT) along with the ever-increasing demands of the mobile Internet, impose high spectral efficiency, low latency and massive connectivity requirements on fifth generation (5G) wireless networks and beyond. Accordingly, next-generation wireless communication systems are expected to support heterogeneous devices for various standards and services with particularly high throughput and low latency requirements. This applies to both large scale and small scale network set ups, which calls for flexible and software reconfigurable transceivers that are capable of supporting the desired quality of service demands. To this end, direct conversion transceivers have attracted considerable attention owing to their suitability for higher levels of integration and their reduced cost and power consumption since they require neither external intermediate frequency filters nor image rejection filters. However, in practical communication scenarios, direct-conversion transceiver architectures inevitably suffer from radio-frequency (RF) front-end related impairments, including in-phase/quadrature-phase imbalances (IQI), which limit the overall system performance. In this context, IQI, which refers to the amplitude and phase mismatch between the I and Q branches of a transceiver, leads to imperfect image rejection, which results in performance degradation of both conventional and emerging communication systems \cite{883502, bernard2001digital}. In ideal scenarios, the I and Q branches of a mixer have equal amplitude and a phase shift of $90\degree$, providing an infinite attenuation of the image band; however, in practice, direct-conversion transceivers are sensitive to certain analog front-end related impairments that introduce errors in the phase shift as well as mismatches between the amplitudes of the I and Q branches which corrupt the down-converted signal constellation, thereby increasing the corresponding error rate \cite{883502}. 

It is recalled that depending on the receiver's (RX) ability to exploit knowledge of the carrier's phase to detect the signals, the detection can be classified into coherent and noncoherent \cite{bernard2001digital}. In the former, exact knowledge of the carrier phase as well as the channel state information (CSI) is required at the receiver, which is a challenging task in certain practical applications. On the contrary, this information is not required in noncoherent detection, which ultimately reduces the corresponding receiver complexity at the expense of a decreased spectral efficiency or a performance penalty. Therefore, the associated complexity-performance tradeoff must be thoroughly quantified in order to optimize the overall system efficiency and performance.  

%I/Q signal processing is part of  modern communication transceivers addressing   the problem of matching the amplitudes and phases of the branches, which result to an interference from the image signal. 
%This issue constitutes, in fact,  a core part in emerging technologies, which are typically characterized  by increased complexity and quality of service requirements  \cite{Hamila_1, 7282087, Hamila_3} and the references therein. 

I/Q signal processing is widely utilized in today's communication transceivers which gives rise to the problem of matching the amplitudes and phases of the branches, resulting in an interference from the image signal. Motivated by this practical concern, several recent works have proposed to model, mitigate or even exploit IQI, see \cite{Hamila_1, 7282087, Hamila_3} and the references therein. Specifically, the authors in \cite {6732973} derive the signal-to-interference-plus-noise-ratio (SINR), taking into account the channel correlation between the subcarriers, in the context of orthogonal frequency division multiplexing (OFDM) systems. Assuming IQI at the receiver only, the SINR probability distribution function (PDF) of generalized frequency division multiplexing under Weibull fading channels was derived and the average symbol error rate (SER) of $M$-ary quadrature amplitude modulation ($M-$QAM) was formulated in \cite{7343345}. For Rayleigh fading channels, the ergodic capacity of OFDM systems with receiver IQI and single-carrier frequency-division-multiple-access (SC-FDMA) systems with joint transmitter (TX)/receiver IQI was investigated in \cite{4533292} and \cite{6736643}, respectively. Likewise, the bit error rate (BER) of differential quadrature phase shift keying (DQPSK) was recently derived in \cite{ICC} for single-carrier and multi-carrier systems in the presence of IQI. Moreover, the authors in \cite{5990406} derived the SER of OFDM with $M-$QAM constellation, over frequency selective channels with RX IQI, whereas the authors in \cite{boulogeorgos2015effects} quantified the effects of IQI on the outage probability of both single-carrier and multi-carrier systems over $N$*Nakagami-$m$ fading conditions. Likewise, the error rate of free-space optical systems using subcarrier intensity modulated QPSK over Gamma-Gamma fading channels with receiver IQI was investigated in \cite{7876107}, while Chen \textit{et al.} recently analyzed the impact of IQI on differential space time block coding (STBC)-based OFDM systems by deriving an error floor and approximations for the corresponding BER \cite{7842134,7782798}. Finally, IQI has also been studied in half-duplex (HD) and full duplex (FD) amplify and forward (AF) and decode and forward (DF) coopeartive systems \cite{Hamila_2, Hamila_4, Hamila_5, Michalis_1}, as well as two-way relay systems and multi-antenna systems \cite{Michalis_2, Michalis_3, Michalis_4}. 

\subsection {Motivation}

It is well known that coherent information detection requires full knowledge of the CSI at the receiver, which is typically a challenging task as sophisticated and often complex channel estimation algorithms are required. In this context, noncoherent detection has been proposed as an efficient technique particularly for low-power wireless systems such as wireless sensor networks and relay networks \cite{abouei2011green}. The main advantage of this scheme stems from the fact that it simplifies the detection since it eliminates the need for channel estimation and tracking, which reduces the cost and complexity of the receiver \cite{1273685,1638651}. However, this comes at a cost of higher error rate or lower spectral efficiency; as a result, selecting the most suitable modulation scheme depends on the considered application and both noncoherent and coherent detection are efficiently implemented in practical systems. Moreover, it is recalled that the detrimental effects of RF front-end impairments on the system performance are often neglected. This also concerns the effects of IQI on \textit{M-}ary phase shift keying (\textit{M-}PSK), \textit{M-}ary differential phase shift keying (\textit{M-}DPSK) and \textit{M-}ary frequency-shift keying (\textit{M}-FSK), which, to the best of the authors' knowledge, have not yet been addressed in the open technical literature. To this end, this article is devoted to the quantification and analysis of these effects in wireless communications over multipath fading channels.

\subsection {Contribution}

The main objective of this paper is to develop a general framework for the comprehensive analysis of coherent and noncoherent modulation schemes under different IQI scenarios. To this end, we consider both single-carrier and multi-carrier systems and we quantify the effects of TX IQI, RX IQI and joint TX/RX IQI for \textit{M}-PSK, \textit{M}-DPSK and \textit{M}-FSK constellations over Rayleigh fading channels. In more details, the main contributions of this work are summarized as follows:
\begin{itemize}
\item We derive novel analytic expressions for the SINR PDF and the cumulative distribution function (CDF) for single-carrier systems over Rayleigh fading channels with TX and/or RX IQI along with a novel generalized closed form expression for the corresponding SINR MGF.
\item We derive novel closed form expressions for the SINR PDF, CDF and MGF for the case of multi-carrier systems over Rayleigh fading channels with TX and/or RX IQI.
\item Using the derived MGFs, we derive the corresponding SER expressions for the cases of \textit{M}-PSK, \textit{M}-DPSK and \textit{M}-FSK constellations. 
\item We derive simple and fairly tight upper bounds for the SER of the different investigated modulation schemes with TX and/or RX IQI, which provide insights into the effect of each parameter on the system performance.
\end{itemize}

\subsection{Organization and Notations}

The remainder of the paper is organized as follows: Section \ref{S:SM} provides a brief overview of the considered modulation schemes. In Section \ref{S:MGF}, the SINR PDF, CDF and MGF is derived for single-carrier and multi-carrier systems with IQI, while Section \ref{sec:PA} presents the SER of \textit{M}-PSK, \textit{M}-DPSK and \textit{M}-FSK with IQI. Upper bounds on the SER of the considered scenarios are derived in Section \ref{S:Asymptotic} whereas the corresponding numerical results and discussions are provided in Section \ref{S:Sim}. Finally, closing remarks are given in Section \ref{S:Conclusions}.

\subsubsection*{Notations}
Unless otherwise stated, $\left(\cdot\right)^{*}$ denotes conjugation and $j=\sqrt{-1}$. The operators $\mathbb{E}\left[\cdot\right]$ and $\left|\cdot\right|$ denote statistical expectation and absolute value operations, respectively. Also, $f_X\left(x \right)$ and $F_X\left(x \right)$ denote the PDF and CDF of $X$, respectively while $\mathcal{M}_X\left(s\right)$ is the MGF associated with $X$. Finally, the subscripts $t/r$ denote the up/down-conversion process at the TX/RX, respectively.

\section{System model}\label{S:SM}

We assume that a signal, $s$, is transmitted over a flat fading wireless channel, $h$, which follows a Rayleigh distribution and is subject to additive white Gaussian noise, $n$. Assuming that the TX/RX are equipped with a single antenna, we first revisit briefly the signal model for the considered \textit{M}-ary PSK, DPSK and FSK modulation schemes. 

\subsection {Coherent Detection of \textit{M}-PSK Symbols}

Assuming \textit{M}-PSK modulation, it is recalled that 
\begin{equation} 
\label{eq:theta_info}
\theta_m= \frac{\left(2m-1\right)\pi}{M}, \qquad m=1,2,\dots,M
\end{equation}
Hence, the complex baseband signal at the transmitter in the $l^{\rm th}$ symbol interval is given by
\begin{equation}
s[l]=A_c e^{j\theta[l]}
\end{equation}

\noindent where $\theta[l]$ is the information phase in the $l^{\rm th}$ symbol. Assuming that the receiver has perfect knowledge of the CSI as well as carrier phase and frequency, the complex baseband signal at the receiver is represented as 
\begin{equation}
x[l]=A_c e^{j\theta[l] }+n[l]. 
\end{equation}

\subsection{Noncoherent Detection of \textit{M}-DPSK Symbols}

Assuming $M$-ary DPSK modulation, the information phase in (\ref{eq:theta_info}) is modulated on the carrier as the difference between two adjacent transmitted phases. Considering that the channel is slowly varying and remains constant over two consecutive symbols, the receiver takes the difference of two adjacent phases to reach a decision on the information phase without knowledge of the carrier phase and channel state \cite{simon2005digital}. In this context, the information phases $\Delta \theta[l]$ are first differentially encoded to a set of phases as follows 
\begin{equation}
\theta[l]=\left(\theta[l-1]+\Delta \theta[l]\right)  \,\textup{mod}\,2\pi
\end{equation}
\noindent where $\Delta \theta_m =(2m-1)\pi/M,\: m=1,...,M$ and $\Delta \theta[l]$ is the information phase in the $l^{\rm th}$ symbol interval. The modulated symbol $s[l]$ is then obtained by applying a phase offset to the previous symbol $s[l-1]$, namely 
\begin{equation}
s[l]=s[l-1]e^{j\theta[l]}
\end{equation}
\noindent where $s[1]=1$. Similarly, the decision variable is obtained from the phase difference between two consecutive received symbols as follows
\begin{equation}
\widehat{s}[l]=r^*[l-1]r[l].
\end{equation}

\subsection{Noncoherent Detection of \textit{M}-FSK Symbols}

Assuming \textit{M}-FSK modulation, the $M$ information frequencies are given by
\begin{equation}
f_m=\left(2 m-1-M \right ){\Delta f}, \qquad m=1,2,\dots,M
\end{equation}
and thus the $l^{\rm th}$ complex baseband symbol at the transmitter is given by 
\begin{equation}
s[l]=A_c e^{j 2\pi f[l]}. 
\end{equation}
The decision variable at the receiver is then obtained by multiplying the received signal by the set of complex sinusoids $e^{j 2 \pi f_m},$ $m=1,2,...,M$ and passing them through $M$ matched filters. For orthogonal signals, the frequency spacing is chosen as $\Delta f= N/T_s$, where $T_s$ is the symbol period and $N$ is an integer. 

\section {MGF of the Received SINR with IQI}
\label{S:MGF}
At the receiver RF front end, the received RF signal undergoes various processing stages including filtering, amplification, and analog I/Q demodulation (down-conversion) to baseband and sampling. Assuming an ideal RF front end, the baseband equivalent received signal is represented as
\begin{align}
r_{\text{id}} = h s + n
\label{Eq:r_ideal} 
\end{align}  
\noindent where $h$ denotes the channel coefficient and $n$ is the circularly symmetric complex additive white Gaussian noise (AWGN) signal. The instantaneous signal to noise ratio (SNR) per symbol at the receiver input is given by 
\begin{equation}
\gamma_{\text{id}} = \frac{E_{s}}{N_{0}} \left|h\right|^{2}
\label{Eq:gamma_ideal}
\end{equation}
\noindent where $E_{s}$ is the energy per transmitted symbol and $N_{0}$ denotes the single-sided AWGN power spectral density.

Likewise, in the case of multicarrier systems, the corresponding baseband equivalent received signal at the $k^{\rm th}$ carrier is represented as
\begin{equation}
r_{\text{id}}\left(k\right) = h\left(k\right) s\left(k\right) + n\left(k\right)
\end{equation}
\noindent where $s\left(k\right)$ is the transmitted signal at the $k^{\rm th}$ carrier, whereas $ h\left(k\right)$ and $n\left(k\right)$ denote the corresponding channel coefficient and the circular symmetric complex AWGN, respectively. Hence, the corresponding instantaneous SNR can be represented as 
\begin{equation}
\gamma_{\text{id}}\left(k\right) = \frac{E_{s}}{N_{o}} \left|h\left(k\right)\right|^{2}. 
\end{equation}

It is assumed that the RF carriers are up/down converted to the baseband by direct conversion architectures, while we assume frequency independent IQI caused by the gain and phase mismatches of the I and Q mixers. In this context, the time-domain baseband representation of the IQI impaired signal is given by \cite{B:Schenk-book} 
\begin{align}
g_{\text{IQI}}= \mu_{t/r} g_{\text{id}} + \nu_{t/r} g_{\text{id}}^{*}
\label{Eq:IQI_basic}
\end{align}
\noindent where $g_{\text{id}}$ is the baseband IQI-free signal and $g_{\text{id}}^{*}$ is due to IQI. In addition, the corresponding IQI coefficients $\mu_{t/r}$ and $\nu_{t/r}$ are given by
\begin{equation}
 \begin{Bmatrix}{\mu_t}\\ {\nu_t}\end{Bmatrix} =  \frac{1 \{ \pm \}  \epsilon_t e^{\{ \pm \} j \phi_t}}{2} 
\end{equation}
and
\begin{equation}
 \begin{Bmatrix} {\mu_r}\\ {\nu_r}\end{Bmatrix} =  \frac{1 \{ \pm \}  \epsilon_r e^{\{ \mp \} j \phi_r}}{2}
\end{equation}
\noindent where $\epsilon_{t/r}$ and $\phi_{t/r}$ denote the TX/RX amplitude and phase mismatch levels, respectively. 
It is noted that for ideal RF front-ends, $\phi_{t/r}=0\degree$ and $\epsilon_{t/r}=1$, which implies that $\mu_{t/r}=1$ and $\nu_{t/r}=0$. Moreover, the TX/RX image rejection ratio ({\rm IRR}) is given by 
\begin{equation}
{\rm IRR}_{t/r}=\frac{\left|\mu_{t/r}\right|^2}{\left|\nu_{t/r}\right|^2}.
\end{equation}

It is recalled that in single-carrier systems, IQI causes distortion to the signal from its own complex conjugate while in multi-carrier systems, IQI causes distortion to the transmitted signal at carrier $k$ from its image signal at carrier $-k$. In the following, assuming that both the transmitter and receiver are equipped with a single antenna, we revisit the signal model of both single-carrier and multi-carrier systems in the presence of IQI at the transmitter and/or receiver. Then, we derive novel analytic expressions for the SINR PDF, CDF and MGF in each scenario.

\subsection{Single-Carrier Systems}
Single-carrier modulation is receiving increasing attention due to its robustness towards RF impairments compared to multi-carrier modulation; see \cite{7416621} and the references therein. Hence, it is considered more suitable for low complexity and low power applications. In what follows, we derive unified closed form expressions for the SINR PDF, CDF and MGF of single-carrier systems in the presence of IQI.

\subsubsection{Signal Model}\label{SSS:SM_SC_sig}

\begin{itemize}
\item \textit{TX IQI and ideal RX: }
This case assumes that the RX RF front-end is ideal, while the TX experiences IQI. Based on this, the baseband equivalent transmitted signal is expressed as
\begin{align}
s_{\text{IQI}} = \mu_t s +\nu_t s^{*}
\label{Eq:s_SC_TX_IQI} 
\end{align}
\noindent whereas the baseband equivalent received signal is given by 
\begin{align}
h s_{\text{IQI}} + n = \mu_t h s + \nu_t h s^{*} + n.
\label{Eq:r_SC_TX_IQI} 
\end{align}
\noindent Hence, the instantaneous SINR per symbol at the input of the receiver is given by
\begin{equation}
\gamma_{\text{IQI}} =\frac{\left|\mu_t \right|^{2}}{\left|\nu_t\right|^{2}+\frac{1}{\gamma_{\text{id}}}}.
\label{Eq:gamma_SC_TX}
\end{equation}
%
%\noindent where $\gamma_{\text{id}} = {E_{s} |h|^2}/{N_{0}}$ is the instantaneous signal to noise ratio (SNR) per symbol at the receiver input, assuming an ideal RF front end.
%
\item \textit{RX IQI and ideal TX: }
This case assumes that the TX RF front-end is ideal, while the RX is subject to IQI. Hence, the baseband equivalent received signal is given by 
\begin{equation}
r_{\text{IQI}} = \mu_r h s + \nu_r h^{*} s^{*} + \mu_r n + \nu_r n^{*}. 
\label{Eq:r_SC_RX_IQI}
\end{equation}
\noindent Therefore, at the RX input, the instantaneous SINR per symbol is expressed as
\begin{equation}
\gamma_{\text{IQI}} =\frac{\left|\mu_r\right|^{2}}{\left|\nu_r\right|^{2}+\frac{\left|\mu_r\right|^{2}+\left|\nu_r\right|^{2}}{\gamma_{\text{id}}}}.
\label{Eq:gamma_SC_RX}
\end{equation}

\item \textit{Joint TX/RX IQI: }
This case assumes that both TX and RX are impaired by IQI and the baseband equivalent received signal is given by
\begin{equation}
r_{\text{IQI}} = \left(\xi_{11} h + \xi_{22} h^{*}\right)s+\left(\xi_{12} h + \xi_{21} h^{*}\right)s^{*}+\mu_r n + \nu_r n^{*}.
\label{Eq:r_SC_TX_RX_IQI}
\end{equation}
\noindent Based on this, the instantaneous SINR per symbol at the RX input is given by
\begin{equation}
\gamma_{\text{IQI}} =\frac{\left|\xi_{11}\right|^2 + \left|\xi_{22}\right|^2 }{ \left|\xi_{12}\right|^{2}+ \left|\xi_{21}\right|^{2} +\frac{\left|\mu_r\right|^{2}+\left|\nu_r\right|^{2}}{\gamma_{\text{id}}}}
\label{Eq:gamma_SC_TXRX}
\end{equation}
\noindent where $\xi_{11} = \mu_r \mu_t$, $\xi_{22} = \nu_r \nu_t^{*}$, $\xi_{12} =\mu_r \nu_t$, and $\xi_{21} =\nu_r \mu_t^{*}$.
\end {itemize}
\subsubsection {SINR Distribution}

From (\ref{Eq:gamma_SC_TX}), (\ref{Eq:gamma_SC_RX}) and (\ref{Eq:gamma_SC_TXRX}), the SINR of single-carrier systems in the presence of IQI can be expressed as 
\begin{equation}
\gamma_{\text{IQI}} =\frac{\alpha}{\beta+\frac{A}{\gamma_{\text{id}}}}
\end{equation}
\noindent where the parameters $\alpha$, $\beta$, and $A$ are given in Table \ref{tab:coeff}.
\begin{table}[H]
\protect\caption{Single-carrier systems impaired by IQI parameters}
\centering{}%
\label{tab:coeff}
\begin{tabular}{|c|c|c|c|}
\cline{2-4}
\multicolumn{1}{c|}{} & $\alpha$ &$\beta$ &  $A$\tabularnewline
\hline 
$\text{TX IQI}$ & $|\mu_t |^2 $&	$|\nu_t |^2 $& $1 $\tabularnewline
\hline 
$\text{RX IQI}$ & $|\mu_r|^2$ &$	|\nu_r |^2$	& $|\mu_r |^2+|\nu_r |^2$ \tabularnewline
\hline 
$\text{Joint TX/RX IQI}$ & $|\xi_{11} |^2+|\xi_{22} |^2$&	$|\xi_{12} |^2+|\xi_{21} |^2	$& $|\mu_r |^2+|\nu_r |^2$ \tabularnewline
\hline 
\end{tabular}
\end{table}
\noindent Hence, the CDF of $\gamma_{\text{IQI}}$ is obtained as
\begin{equation}
\label{eq:SC_cdf}
F_{\gamma_{IQI}}\left(x\right )=F_{\gamma_{id}}\left(\frac{A}{\frac{\alpha}{x} -\beta}\right ) 
\end{equation}
\noindent where $\gamma_{id}$ is the IQI free SNR, which follows an exponential distribution with CDF and PDF given by
\begin{equation}
\label{eq:CDF_ideal}
F_{\gamma_{id}}\left(x\right)=1-\exp\left(-\frac{x}{\overline{\gamma}}\right)
\end{equation}
\noindent and
\begin{equation}
\label{eq:PDF_ideal}
f_{\gamma_{id}}\left(x\right)=\frac{\exp\left(-\frac{x}{\overline{\gamma}}\right)}{\overline{\gamma}} 
\end{equation}
\noindent respectively, where $\overline \gamma=E_s/N_0$ is the average SNR. Hence, assuming TX and/or RX IQI, the corresponding SINR CDF is given by
\begin{equation}
F_{\gamma_{\rm IQI}}\left(x\right )= 1-e^{-\frac{A}{\overline{\gamma} \left(\frac{\alpha}{x} -\beta\right )}}, \qquad \hspace{3em} 0\leq x\leq \frac{\alpha}{\beta} 
\end{equation}
\noindent Given that $f_{\gamma_{\rm IQI}}\left(x\right )=\frac{{\rm d}}{{\rm d} x}F_{\gamma_{\rm IQI}}\left(x\right )$, the SINR PDF, in the presence of IQI, is given by
\begin{equation}
\label{eq:SC_pdf}
f_{\gamma_{\rm IQI}}\left(x \right)=\frac{\alpha Ae^{-\frac{A}{\overline{\gamma} \left(\frac{\alpha}{x} -\beta\right )}}}{\overline{\gamma}\left(\alpha-x \beta \right )^2}  
\end{equation} 
\noindent which is valid for $0\leq x \leq \frac{\alpha}{\beta}$.

\subsubsection{Moment Generating Function (MGF)}
The MGF is an important statistical metric and constitutes a convenient tool in digital communication systems over fading channels \cite{simon2005digital}. In what follows, we derive a generalized closed form expression for the SINR MGF of single-carrier systems in the presence of IQI, which will be particularly useful in the subsequent analysis. 
\begin{proposition}
For single-carrier systems impaired by IQI, the MGF of the instantaneous fading SINR is given by
\begin{equation}
\label{eq:SC_mgf}
\mathcal{M}_{\gamma_{\rm IQI}}\left(s\right)= e^{\frac{\alpha}{\beta}s+\frac{A}{\beta\overline{\gamma}}}\Gamma\left(1,\frac{A}{\overline{\gamma}\beta} ;\frac{s  \alpha A}{\beta^2\overline{\gamma}}\right )
\end{equation}
\noindent where $\Gamma\left(\alpha,x;b\right)=\int_{x}^{\infty}t^{\alpha-1}e^{-t-\frac{b}{t}} \mathrm{d}t$ is the extended upper incomplete Gamma function \cite{chaudhry2001class}.
\end{proposition}
\begin{proof}
By recalling that \cite{simon2005digital}
\begin{equation}
\label{eq:def_mgf}
\mathcal{M}_{\gamma_{\rm IQI}}\left(s\right)=\int_{0}^{\infty}e^{s x}f_{\gamma_{\rm IQI}}\left(x \right ) {\rm d}x
\end{equation}
\noindent and substituting (\ref{eq:SC_pdf}) into (\ref{eq:def_mgf}) yields
\begin{equation}
\label{eq:mgf_TXIQI1}
\mathcal{M}_{\gamma_{\rm IQI}}\left(s\right)=\int_{0}^{\frac{\alpha}{\beta}}e^{sx}\frac{\alpha Ae^{-\frac{A}{\overline{\gamma} \left(\frac{\alpha}{x} -\beta\right )}}}{\overline{\gamma}\left(\alpha-x\beta \right )^2} {\rm d} x.  
\end{equation}
\noindent By also considering the change of variable $y={\alpha} -{\gamma}\beta$ and after some mathematical manipulations, one obtains 
\begin{equation}
\label{eq:mgf_TXIQI3}
\mathcal{M}_{\gamma_{\rm IQI}}\left(s\right)=\frac{\alpha A}{\overline{\gamma}\beta} e^{\frac{\alpha}{\beta}s+\frac{A}{\beta\overline{\gamma}}}\int_{0}^{\alpha}\, e^{-\frac{s y}{\beta}-\frac{\alpha A}{\beta\overline{\gamma}y}} {\rm d}y. 
\end{equation}
\noindent Based on this and by taking $z=\frac{\alpha A}{\beta\overline{\gamma}y}$, equation (\ref{eq:SC_mgf}) is deduced, which completes the proof. 
\end{proof}

\subsection{Multi-carrier systems }\label{S:MC}

It is recalled that multi-carrier systems divide the signal bandwidth among $K$ carriers, which provides several advantages including enhanced robustness against multipath fading. Based on this, Long-Term Evolution (LTE) employs orthogonal frequency division multiplexing (OFDM) in the downlink. In this subsection, we derive the SINR PDF, CDF and MGF of multi-carrier systems in the presence of IQI, which creates detrimental performance effects. To this end, we assume that the RF carriers are down converted to the baseband by wideband direct conversion. We also denote the set of signals as $S=\{-\frac{K}{2},\dots,-1,1,\dots,\frac{K}{2}\}$ and assume that there is a data signal present at the image subcarrier and that the channel responses at the $k^{\rm th}$ carrier and its image are uncorrelated. %In addition, it is noted that the correlation between the channel responses at the $k\rm{th}$ carrier and its image is small due to their large spectral separation. To this effect, it is realistic to assume them statistically independent.

\subsubsection{Joint TX/RX impaired by IQI}\label{SSS:SM_MC_TX_RX_IQI}

Here, we consider the general scenario where both the TX and RX suffer from IQI. The baseband equivalent received signal in this case is given by
\begin{equation}
r_{\text{IQI}} = \left(\xi_{11} h\left(k \right ) + \xi_{22} h^{*}\left(-k \right )\right)s\left(k \right )+ \left(\xi_{12} h\left(k \right ) + \xi_{21} h^{*}\left(-k \right )\right)s^{*}\left(-k \right )+\mu_r n\left(k \right ) + \nu_r n^{*}\left(-k \right )
\label{Eq:r_MC_TX_RX_IQI}
\end{equation}
\noindent where the carrier $-k$ is the image of the carrier $k$. To this effect, the instantaneous SINR per symbol at the input of the RX is given by 
\begin{equation}
\gamma =\frac{\left|\xi_{11}\right|^2 + \left|\xi_{22}\right|^2 \frac{{\gamma_{\text{id}}}\left(-k \right )}{{\gamma_{\text{id}}}\left(k \right )}}{ \left|\xi_{12}\right|^{2}+ \left|\xi_{21}\right|^{2} \frac{\gamma_{\text{id}}\left(-k \right )}{{\gamma_{\text{id}}}\left(k \right )}+\frac{\left|\mu_r\right|^{2}+\left|\nu_r\right|^{2}}{{\gamma_{\text{id}}}\left(k \right )}}
\label{eq:gamma_MC_TXRX}
\end{equation}
\noindent where
\begin{equation}
\gamma_{\text{id}}\left(-k \right ) = \frac{E_{s}}{N_{0}} \left|h\left(-k \right )\right|^{2}. 
\end{equation}
Therefore, for the case of given $\gamma_{id}\left(-k\right)$ and with the aid of (\ref{eq:gamma_MC_TXRX}) and (\ref{eq:CDF_ideal}), the conditional SINR CDF can be expressed as
\begin{equation}
\label{eq:MC_JOINT_CDF1}
F_{\gamma_{\rm IQI}}\left(x |\gamma_{id}\left(-k \right )\right)=1-\exp\left(\frac{-x\left(\left|\xi_{21}\right|^{2}\gamma_{id}\left(-k \right )+\left|\mu_r\right|^{2}+\left|\nu_r\right|^{2} \right )-\left|\xi_{22}\right|^{2}\gamma_{id}\left(-k \right )}{\overline{\gamma}\left( \left|\xi_{11}\right|^{2}-x \left|\xi_{12}\right|^{2} \right )}\right ).
\end{equation}
\noindent Based on this, the unconditional CDF is obtained by integrating (\ref{eq:MC_JOINT_CDF1}) over (\ref{eq:PDF_ideal}), yielding
\begin{equation}
\label{eq:MC_JOINT_CDF}
F_{\gamma_{\rm IQI}}\left(x \right )=1-\frac{\exp\left(-\frac{x \left(\left|\mu_r\right|^{2}+\left|\nu_r\right|^{2}\right)}{\overline{\gamma}\left( \left|\xi_{11}\right|^{2}-x \left|\xi_{12}\right|^{2}\right )} \right )}{1+\frac{x\left|\xi_{21}\right|^{2}-\left|\xi_{22}\right|^{2}}{\left( \left|\xi_{11}\right|^{2}-x\left|\xi_{12}\right|^{2}\right )}},\qquad 0\leq x\leq \frac{\left|\xi_{11}\right|^2}{\left|\xi_{12} \right|^2}
\end{equation}
\noindent whereas the SINR PDF is obtained as
\begin{equation}
\label{eq:MC_JOINT_PDF}
f_{\gamma_{\rm IQI}}\left(x \right )=\frac{\exp\left(-\frac{x\left(|\mu_r|^{2}+|\nu_r|^{2}\right)}{\overline{\gamma}\left( |\xi_{11}|^{2}-x|\xi_{12}|^{2}\right )} \right )\left( \frac{|\xi_{11}|^{2}\left(|\mu_{R}|^{2} +\left|\nu_{R}\right|^{2}\right )}{\overline{\gamma}}+\frac{|\xi_{21}|^{2}|\xi_{11}|^{2}-|\xi_{12}|^{2}|\xi_{22}|^{2}}{1+\frac{x|\xi_{21}|^{2}-\left|\xi_{22}\right|^{2}}{ |\xi_{11}|^{2}-x|\xi_{12}|^{2}}}\right )}{\left( |\xi_{11}|^{2}-x|\xi_{12}|^{2}\right )\left(|\xi_{11}|^{2}-|\xi_{22}|^{2} +x\left(|\xi_{21}|^{2} -|\xi_{12}|^{2}\right )\right )},
\end{equation}
\noindent which is valid for $0\leq x \leq  |\xi_{11}|^2/|\xi_{12}|^2$.

\begin{proposition}
The MGF of multi-carrier systems impaired by joint TX/RX IQI is given by %(\ref{eq:MGF_MC_joint_eq}), (\ref{eq:MGF_MC_joint_lt}) and (\ref{eq:MGF_MC_joint_gt})
\begin{equation}
\label{eq:MGF_MC_joint_eq}
\mathcal{M_{\gamma_{\rm IQI}}}\left(s\right)=C+\frac{|\xi_{12}|^2}{s\left(|\xi_{11}|^2-|\xi_{22}|^2 \right )}e^{s\frac{|\xi_{11}|^2}{|\xi_{12}|^2}+\frac{|\mu_{r}|^{2}+|\nu_{r}|^{2}}{|\xi_{12}|^2 \overline{\gamma}}}\gamma\left(2,s\frac{|\xi_{11}|^2}{|\xi_{12}|^2};s\frac{|\xi_{11}|^2\left( |\mu_{r}|^{2}+|\nu_{r}|^{2}\right )}{|\xi_{12}|^4\overline{\gamma}} \right)
\end{equation}
\noindent for $|\xi_{12}|^2=|\xi_{21}|^2$,
\begin{equation}
\begin{split}
\label{eq:MGF_MC_joint_lt}
\mathcal{M}_{\gamma_{\rm IQI}}\left(s\right)=\frac{|\xi_{11}|^2}{|\xi_{11}|^2-|\xi_{22}|^2}+ & \sum_{k=0}^{\infty }\frac{\left(-1\right)^ks^k \, d^k e^{\frac{|\mu_{R}|^{2}+|\nu_{R}|^{2}}{|\xi_{12}|^2 \overline{\gamma}}+s \frac{|\xi_{11}|^2}{|\xi_{12}|^2}}}{\left( |\xi_{12}|^2-|\xi_{21}|^2\right )^{k+1}|\xi_{12}|^{2k-2}}\\ & \times \gamma\left(1-k,s\frac{|\xi_{11}|^2}{|\xi_{12}|^2};s\frac{|\xi_{11}|^2\left( |\mu_{r}|^{2}+|\nu_{r}|^{2}\right )}{|\xi_{12}|^4\overline{\gamma}} \right)
\end{split}
\end{equation}
\noindent for $\left|\frac{{|\xi_{11}|^2|\xi_{21}|^2-|\xi_{22}|^2|\xi_{12}|^2}}{|\xi_{12}|^2-|\xi_{21}|^2}\right|< |\xi_{11}|^2$, and
\begin{equation}
\begin{split}
\label{eq:MGF_MC_joint_gt}
\mathcal{M_{\gamma_{\rm IQI}}}\left(s\right)=C+e^{s\frac{|\xi_{11}|^2}{|\xi_{12}|^2}+\frac{|\mu_{r}|^{2}+|\nu_{r}|^{2}}{|\xi_{12}|^2 \overline{\gamma}}} & \sum_{k=0}^{\infty } \frac{\left(-1 \right )^k\left(|\xi_{12}|^2 -|\xi_{21}|^2\right)^{k }  |\xi_{12}|^{2k+4}}{d^{k+1}s^{k+1}}\\ & \times\gamma\left(k+2,s\frac{|\xi_{11}|^2}{|\xi_{12}|^2};s\frac{|\xi_{11}|^2\left( |\mu_{r}|^{2}+|\nu_{r}|^{2}\right )}{|\xi_{12}|^4\overline{\gamma}}\right )
\end{split}
\end{equation}
\noindent for $\left|\frac{{|\xi_{11}|^2|\xi_{21}|^2-|\xi_{22}|^2|\xi_{12}|^2}}{|\xi_{12}|^2-|\xi_{21}|^2}\right|>  |\xi_{11}|^2$, where $\gamma\left(\alpha,x;b\right)=\int_{0}^{x}t^{\alpha-1}e^{-t-\frac{b}{t}} \mathrm{d}t$ is the extended lower incomplete Gamma function \cite{chaudhry2001class}, while 
\begin{equation}
\label{eq:const}
C=\frac{|\xi_{11}|^2}{|\xi_{11}|^2-|\xi_{22}|^2}
\end{equation}
\noindent and
\begin{equation}
\label{eq:d}
d={|\xi_{11}|^2|\xi_{21}|^2-|\xi_{22}|^2|\xi_{12}|^2}.
\end{equation}
\end{proposition}

\begin{proof}
The proof is provided in Appendix \ref{Appendix:MGF_MC_JOINT}. 
\end{proof}

\subsubsection{TX Impaired by IQI}\label{SSS:SM_MC_TX_IQI}
Assuming that the RX RF front-end is ideal, while the TX experiences IQI, the baseband equivalent received signal is
\begin{align}
s_{\text{IQI}} = \mu_t s\left(k\right)h\left(k\right) + \nu_t s^{*}\left(-k\right)h\left(k\right)+n\left(k\right)
\label{Eq:s_MC_TX_IQI} 
\end{align}
\noindent and the instantaneous SINR per symbol at the input of the RX is given by 
\begin{equation}
\begin{split}
\gamma_{\text{IQI}} =\frac{\left|\mu_t \right|^{2}}{\left|\nu_t\right|^{2}+\frac{1}{\gamma_{\text{id}}\left(k\right)}}.
\label{Eq:gamma_MC_TX}
\end{split}
\end{equation}
\noindent Hence, by setting $\mu_r=1$ and $\nu_r =0$ in (\ref{eq:MC_JOINT_CDF}), it follows that 
\begin{equation}
F_{\gamma_{\rm IQI}}\left(x\right )= 1-e^{-\frac{1}{\overline{\gamma} \left(\frac{|\mu_t |^2}{x} -|\nu_t |^2\right )}}, \qquad  0\leq x \leq \frac{|\mu_t |^2}{|\nu_t |^2}
\end{equation}
\noindent which yields straightforwardly the corresponding SINR PDF, namely 
\begin{equation}
\label{eq:MC_pdf}
f_{\gamma_{\rm IQI}}\left(x\right)=\frac{|\mu_t |^2 e^{-\frac{1}{\overline{\gamma} \left(\frac{|\mu_t |^2}{x} -|\nu_t |^2\right )}}}{\overline{\gamma}\left(|\mu_t |^2-x|\nu_t |^2 \right )^2}  
\end{equation} 
\noindent which is valid for $0\leq x \leq  |\mu_t |^2/|\nu_t |^2$. It is noted that (\ref{eq:MC_pdf}) is similar to (\ref{eq:SC_pdf}) for $\alpha=|\mu_t |^2$, $\beta=|\nu_t |^2$, and $A=1$. Hence, with the aid of (\ref{eq:SC_mgf}), the instantaneous SINR MGF of multi-carrier systems experiencing TX IQI only is given by
\begin{equation}
\label{eq:MC_mgf}
\mathcal{M}_{\gamma_{\rm IQI}}\left(s\right)= e^{\frac{|\mu_t |^2}{|\nu_t |^2}s+\frac{1}{|\nu_t |^2\overline{\gamma}}}\Gamma\left(1,\frac{1}{\overline{\gamma}|\nu_t |^2} ;\frac{s  |\mu_t |^2 }{|\nu_t |^4\overline{\gamma}}\right).
\end{equation}

\subsubsection{RX Impaired by IQI}\label{SSS:SM_MC_RX_IQI}
Assuming that the TX RF front-end is ideal, while the RX is impaired by IQI, the baseband equivalent received signal is represented as
\begin{equation}
r_{\text{IQI}} = \mu_r h\left(k\right) s\left(k\right) + \nu_r h^{*}\left(-k\right) s^{*}\left(-k\right) + \mu_r n\left(k\right) + \nu_r n^{*}\left(-k\right).
\label{Eq:r_MC_RX_IQI}
\end{equation}

Likewise, the instantaneous SINR per symbol at the input of the RX is expressed as
\begin{equation}
\gamma_{\text{\rm IQI}} =\frac{|\mu_r|^{2}}{ |\nu_r|^{2}\frac{\gamma_{\text{id}}\left(-k \right )}{\gamma_{\text{id}}\left(k \right )}+\frac{|\mu_r|^{2}+|\nu_r|^{2}}{\gamma_{\text{id}}\left(k \right )}}. 
\label{Eq:gamma_MC_RX}
\end{equation}
\noindent Hence, substituting $\mu_t=1$ and $\nu_t =0$ in (\ref{eq:MC_JOINT_CDF}) one obtains 
\begin{equation}
\label{eq:MC_RX_CDF}
F_{\gamma_{\rm IQI}}\left(x\right)=1-\frac{|\mu_r|^{2}}{|\mu_r|^{2}+x|\nu_r|^{2}}e^{-\frac{x}{\overline{\gamma}}\left(1+\frac{|\nu_r|^{2}}{|\mu_r|^{2}}\right)},\qquad  0\leq x \leq\infty
\end{equation}
\noindent which with the aid of (\ref{eq:MC_JOINT_PDF}) and after some algebraic manipulations yields the respective SINR PDF, namely 
\begin{equation}
\label{eq:MC_RX_PDF}
f_{\gamma_{\rm IQI}}\left(x \right )=\frac{\exp\left(-\frac{x}{\overline{\gamma}} \left(\frac{|\nu_r|^{2}}{|\mu_r|^{2}} +1\right )\right )}{\overline{\gamma}\left(\frac{x |\nu_r|^{2}}{\left|\mu_r\right|^{2}}+1 \right )}\left(1+\frac{|\nu_r|^{2}}{|\mu_r|^{2}} +\frac{\overline{\gamma}|\nu_r|^{2}}{|\mu_r|^{2}\left(\frac{x|\nu_r|^{2}}{|\mu_r|^{2}}+1 \right )}\right)
\end{equation}
\noindent which is valid for $0\leq x \leq \infty$.

Finally, from (\ref{eq:def_mgf}) and (\ref{eq:MC_RX_PDF}), the corresponding MGF is obtained as 
\begin{equation}
\label{eq:MC_MX_MGF_INT}
\mathcal{M}_{\gamma_{\rm IQI}}\left(s\right)=\frac{1+\frac{|\mu_r|^2}{|\nu_r|^2}}{\overline{\gamma}}\int_{0}^{\infty}\frac{e^{-x\left(\frac{1}{\overline{\gamma}}+\frac{|\nu_r|^2}{\overline{\gamma}|\mu_r|^2}-s\right )}}{x+\frac{|\mu_r|^2}{|\nu_r|^2}}dx+\int_{0}^{\infty}\frac{e^{-x\left(\frac{1}{\overline{\gamma}}+\frac{|\nu_r|^2}{\overline{\gamma}|\mu_r|^2}-s\right )}}{\left(x+\frac{|\mu_r|^2}{|\nu_r|^2}\right )^2} {\rm d}x
\end{equation}
\noindent which with the aid of \cite[eq. (3.352)]{B:Gra_Ryz_Book} and \cite[eq. (3.353)]{B:Gra_Ryz_Book}, eq. (\ref{eq:MC_MX_MGF_INT}) can be expressed by the following closed-form representation  
\begin{equation}
\label{eq:MC_RX_MGF}
\mathcal{M}_{\gamma_{\rm IQI}}\left(s\right)=1-s\frac{|\mu_r|^2}{|\nu_r|^2}e^{\frac{1}{\overline{\gamma}}+\frac{|\mu_r|^2}{\overline{\gamma}|\nu_r|^2}-\frac{s|\mu_r|^2}{|\nu_r|^2}}\textup{Ei}\left(- \frac{1}{\overline{\gamma}}-\frac{|\mu_r|^2}{\overline{\gamma}|\nu_r|^2}+\frac{s|\mu_r|^2}{|\nu_r|^2}\right )  
\end{equation}
\noindent where $\textup{Ei} \left(z\right)=-\int_{-z}^{\infty}{e^{-t}}/{t}  \mathrm{d}t$ denotes the exponential integral function \cite{B:Gra_Ryz_Book}.

The different MGF expressions derived are summarized in Table \ref{table:MGFs}, where $\Lambda=|\mu_{r}|^{2}+|\nu_{r}|^{2}$. It is noted that with the aid of the derived MGFs, the SER of various $M$-ary modulation schemes under different IQI effects as well as multi-channel reception schemes can be readily determined.

\begin{table}[]
\centering
\caption{SINR MGFs}
\label{table:MGFs}
{\renewcommand{\arraystretch}{2}
\begin{tabular}{|l|ll|cl|}
\hline
\multicolumn{1}{|c|}{} & \multicolumn{2}{c|}{Single-carrier systems} & \multicolumn{2}{c|}{Multi-carrier systems} \\ \hline
\multicolumn{1}{|c|}{TX IQI} & \multicolumn{4}{c|}{$\mathcal{M}_{\gamma_{\rm IQI}}\left(s\right)= e^{\frac{|\mu_t |^2}{|\nu_t |^2}s+\frac{1}{|\nu_t |^2\overline{\gamma}}}\Gamma\left(1,\frac{1}{\overline{\gamma}|\nu_t |^2} ;\frac{s,|\mu_t |^2 }{|\nu_t |^4\overline{\gamma}}\right )$} \\ \hline
\multicolumn{1}{|c|}{\multirow{2}{*}{RX IQI}} &\multicolumn{2}{l|}{\multirow{2}{*}{$\mathcal{M}_{\gamma_{\rm IQI}}\left(s\right)=e^{\frac{|\mu_r|^2 s}{|\nu_r|^2}+\frac{\Lambda}{|\nu_r |^2\overline{\gamma}}}\Gamma\left(1,\frac{\Lambda}{\overline{\gamma}|\nu_r |^2} ;\frac{s|\mu_r|^2 \Lambda}{|\nu_r |^4\overline{\gamma}}\right )$}}  & \multicolumn{1}{r}{$\mathcal{M}_{\gamma_{\rm IQI}}\left(s\right)=$} & $1-s\frac{|\mu_r|^2}{|\nu_r|^2}e^{\frac{1}{\overline{\gamma}}+\frac{|\mu_r|^2}{\overline{\gamma}|\nu_r|^2}-\frac{s|\mu_r|^2}{|\nu_r|^2}}$ \\
\multicolumn{1}{|c|}{} & \multicolumn{2}{r|}{ } & \multicolumn{2}{c|}{$\times \textup{Ei}\left(- \frac{1}{\overline{\gamma}}-\frac{|\mu_r|^2}{\overline{\gamma}|\nu_r|^2}+\frac{s|\mu_r|^2}{|\nu_r|^2}\right )$} \\ \hline
\multicolumn{1}{|c|}{} & \multicolumn{1}{r}{} &  & \multicolumn{1}{r}{$\mathcal{M}_{\gamma_{\rm IQI}}\left(s\right)=$} & $C+\frac{|\xi_{12}|^2}{s\left(|\xi_{11}|^2-|\xi_{22}|^2 \right )}e^{s\frac{|\xi_{11}|^2}{|\xi_{12}|^2}+\frac{\Lambda}{|\xi_{12}|^2 \overline{\gamma}}} $ \\
 & \multicolumn{2}{l|}{} & \multicolumn{2}{r|}{$ \times \gamma\left(2,s\frac{|\xi_{11}|^2}{|\xi_{12}|^2};s\frac{|\xi_{11}|^2 \Lambda}{|\xi_{12}|^4\overline{\gamma}} \right), $} \\
 &  &  & \multicolumn{2}{c|}{for $|\xi_{12}|^2=|\xi_{21}|^2$} \\ \cline{4-5} 
Joint IQI & $\mathcal{M}_{\gamma_{\rm IQI}}\left(s\right)=$ & $e^{\frac{|\xi_{11} |^2+|\xi_{22} |^2}{|\xi_{12} |^2+|\xi_{21} |^2}s+\frac{|\mu_r |^2+|\nu_r |^2}{\left(|\xi_{12} |^2+|\xi_{21} |^2\right)\overline{\gamma}}}$ & \multicolumn{1}{l}{$\mathcal{M}_{\gamma_{\rm IQI}}\left(s\right)=$} & $C+\sum_{k=0}^{\infty }\frac{\left(-s\right)^k\, d^k e^{s\frac{|\xi_{11}|^2}{|\xi_{12}|^2}} }{\left( |\xi_{12}|^2-|\xi_{21}|^2\right )^{k+1}|\xi_{12}|^{2k-2}}$ \\
\multicolumn{1}{|c|}{} & \multicolumn{2}{r|}{$\times \Gamma\left(1,\frac{\Lambda}{\overline{\gamma}\left(|\xi_{12} |^2+|\xi_{21} |^2\right)} ;\frac{s\left(|\xi_{11} |^2+|\xi_{22} |^2\right) \Lambda }{\left(|\xi_{12} |^2+|\xi_{21} |^2\right)^2\overline{\gamma}}\right )$} & \multicolumn{2}{r|}{$\times e^{\frac{\Lambda}{|\xi_{12}|^2 \overline{\gamma}}} \gamma\left(1-k,s\frac{|\xi_{11}|^2}{|\xi_{12}|^2};s\frac{|\xi_{11}|^2\Lambda}{|\xi_{12}|^4\overline{\gamma}} \right), $} \\
 & \multicolumn{2}{l|}{} & \multicolumn{2}{c|}{for $\left|\frac{{|\xi_{11}|^2|\xi_{21}|^2-|\xi_{22}|^2|\xi_{12}|^2}}{|\xi_{12}|^2-|\xi_{21}|^2}\right|< |\xi_{11}|^2$} \\ \cline{4-5} 
 &  &  & \multicolumn{1}{l}{$\mathcal{M}_{\gamma_{\rm IQI}}\left(s\right)=$} & $C+ \sum_{k=0}^{\infty } \frac{|\xi_{12}|^{2k+4}e^{s\frac{|\xi_{11}|^2}{|\xi_{12}|^2}}}{\left(|\xi_{21}|^2 -|\xi_{12}|^2\right)^{-k }d^{k+1}s^{k+1}}$ \\
 &  &  & \multicolumn{2}{r|}{$\times e^{\frac{\Lambda}{|\xi_{12}|^2 \overline{\gamma}}} \gamma\left(k+2,s\frac{|\xi_{11}|^2}{|\xi_{12}|^2};s\frac{|\xi_{11}|^2\Lambda}{|\xi_{12}|^4\overline{\gamma}}\right ),$} \\
 &  &  & \multicolumn{2}{c|}{for $\left|\frac{{|\xi_{11}|^2|\xi_{21}|^2-|\xi_{22}|^2|\xi_{12}|^2}}{|\xi_{12}|^2-|\xi_{21}|^2}\right|>|\xi_{11}|^2$} \\ \hline
\end{tabular}}
\end{table}

\section {Symbol Error Rate Analysis}
\label{sec:PA}
This section capitalizes on the derived MGF representation and evaluates the SER performance of both single-carrier and multi-carrier systems employing different coherent and non-coherent \textit{M}-ary modulation schemes in the presence of IQI and multipath fading.

\subsection {Coherent $M$-PSK Symbol Error Rate Analysis}
\label{ss:TX_PSK_BER}

For coherently detected \textit{M}-PSK, the SER under AWGN is given by \cite[eq. (8.22)]{simon2005digital}
\begin{equation}
\label{eq:PSK}
P_{s,{\rm PSK}}=\frac{1}{\pi}\int_{0}^{\frac{\left(M-1\right)\pi}{M}}{\exp\left(-\gamma\frac{g_{{\rm PSK}}}{\sin^2\left(\theta\right)} \right ) {\rm d}\theta}
\end{equation}
\noindent where $\gamma=E_s/N_0$ and $g_{{\rm PSK}}=\sin^2\left(\frac{\pi}{M} \right )$. Under slow fading conditions, the average SER is obtained by averaging (\ref{eq:PSK}) over the considered channel's SINR PDF, namely
\begin{equation}
\label{eq:PSK2}
P_{s,{\rm PSK}}=\frac{1}{\pi}{\int_{0}^{\infty}\int_{0}^{\frac{\left(M-1\right)\pi}{M}}{\exp\left(-x\frac{g_{{\rm PSK}}}{\sin^2\left(\theta\right)} \right )f_\gamma\left(x \right ) {\rm d}\theta}{\rm d}x}
\end{equation}
\noindent which is equivalent to
\begin{equation}
\label{eq:PSK_mgf}
P_{s,{\rm PSK}}=\frac{1}{\pi}\int_{0}^{\frac{\left(M-1\right)\pi}{M}}\mathcal{M}_{\gamma_{\rm IQI}}\left(-x\frac{g_{{\rm PSK}}}{\sin^2\left(\theta\right)}\right) {\rm d}\theta.  
\end{equation}
\noindent Therefore, by assuming PSK modulation, the average SER in the presence of IQI is obtained by substituting the derived MGF expressions into (\ref{eq:PSK_mgf}), which for single-carrier systems is given by
\begin{equation}
\label{eq:SER_PSK}
P_{s,{\rm PSK}}=\frac{1}{\pi}\int_{0}^{\frac{\left(M-1\right)\pi}{M}} e^{-\frac{g_{{\rm PSK}}\alpha}{\sin^2\left(\theta\right)\beta}+\frac{A}{\beta\overline{\gamma}}} \Gamma\left(1,\frac{A}{\overline{\gamma}\beta} ,-\frac{  g_{{\rm PSK}}\alpha A}{\sin^2\left(\theta\right)\beta^2\overline{\gamma}},1\right ) {\rm d}\theta
\end{equation}

\subsection {Differential $M$-PSK Symbol Error Rate Analysis}
\label{ss:TX_Diff_BER}
Considering differential detection of \textit{M}-PSK under AWGN, the exact SER is given by \cite[eq. (8.90)]{simon2005digital}, namely
\begin{equation}
\label{eq:DPSK}
P_{s,{\rm DPSK}}=\frac{1}{ \pi}\int_{0}^{\frac{\left(M-1 \right )\pi}{M}}\exp\left(-\gamma\frac{g_{{\rm PSK}}}{1+\rho\cos\left(\theta \right )}\right)   {\rm d} \theta 
\end{equation}
\noindent where $\rho=\sqrt{1-g_{{\rm PSK}}}$. Based on this and assuming Rayleigh fading conditions, and TX and/or RX IQI, the above expression can be expressed as 
\begin{equation}
\label{eq:dpsk_mgf}
P_{s,{\rm DPSK}}=\frac{1}{\pi}\int_{0}^{\frac{\left(M-1\right)\pi}{M}}\mathcal{M}_{\gamma_{\rm IQI}}\left(-x \frac{g_{{\rm PSK}}}{1+\rho\cos\left(\theta \right )}\right) {\rm d}\theta.
\end{equation}
\noindent The average symbol error rate for \textit{M}-DPSK over Rayleigh fading channels in the presence of IQI is obtained by substituting the derived MGF expressions in (\ref{eq:dpsk_mgf}), which for multi-carrier systems with TX IQI only is given by
\begin{equation}
\label{eq:SER_DPSK}
P_{s,{\rm DPSK}}=\frac{1}{ \pi} \int_{0}^{\frac{\left(M-1 \right )\pi}{M}}e^{-\frac{|\mu_t |^2 g_{{\rm PSK}}}{|\nu_t |^2\left(1+\rho\cos\left(\theta \right )\right)}+\frac{1}{|\nu_t |^2\overline{\gamma}}}  \Gamma\left(1,\frac{1}{\overline{\gamma}|\nu_t |^2} ,\frac{-g_{{\rm PSK}} |\mu_t |^2}{\left(1+\rho\cos\left(\theta \right )\right)|\nu_t |^4\overline{\gamma}},1\right )  {\rm d}\theta. 
 \end{equation}

\subsection {Noncoherent $M$-FSK Symbol Error Rate Analysis}
\label{ss:TX_FSK_BER}
Assuming noncoherent detection of orthogonal signals, corresponding to a minimum frequency spacing $\Delta f= 1/T_s$, the SER of \textit{M}-FSK under AWGN is given by \cite[eq. (8.66)]{simon2005digital}, namely
\begin{equation}
\label{eq:FSK}
P_{s,{\rm FSK}}=\sum_{k=1}^{M-1}\left(-1 \right )^{k+1}\binom{M-1}{k}\frac{\exp\left(-\gamma\frac{k}{k+1} \right )}{k+1}
\end{equation}
\noindent which under fading conditions is expressed as follows
\begin{equation}
\label{eq:FSK_mgf}
P_{s,{\rm FSK}}=\sum_{k=1}^{M-1}\frac{\left(-1 \right )^{k+1}}{k+1}\binom{M-1}{k}\mathcal{M}_{\gamma_{\rm IQI}}\left(-x\frac{k}{k+1} \right ).
\end{equation}
\noindent Therefore, substituting the derived MGF expressions in (\ref{eq:FSK_mgf}) yields the average SER in the presence of IQI, which for the case of multi-carrier systems with RX IQI only is given by 
\begin{equation}
P_{s,{\rm FSK}}=\sum_{k=1}^{M-1}\frac{\left(-1 \right )^{k+1}}{k+1}\binom{M-1}{k}\Bigg[ 1 +\frac{k|\mu_r|^2\textup{Ei}\left(- \frac{|\nu_r|^2+|\mu_r|^2}{\overline{\gamma}|\nu_r|^2}-\frac{k|\mu_r|^2}{\left(k+1\right )|\nu_r|^2}\right )}{e^{-\frac{|\nu_r|^2+|\mu_r|^2}{\overline{\gamma}|\nu_r|^2}-\frac{k|\mu_r|^2}{\left(k+1\right )|\nu_r|^2}}\left(k+1\right )|\nu_r|^2}\Bigg ].  
\end{equation}

To the best of the authors' knowledge, the derived analytic expressions have not been previously reported in the open technical literature. 

\section {Asymptotic Analysis}
\label{S:Asymptotic}

In this section, we analyze the performance of both single-carrier and multi-carrier systems in the asymptotic regime by deriving SER upper bounds. Moreover, since IQI results in interference from either the signal's conjugate or the signal at the image subcarrier, increasing the transmit SNR also increases the interference. Hence, we study the asymptotic behaviors of the derived bounds which provides useful insights into the system behavior.

\subsection {Single-Carrier Systems}
\label{subsec:asym_SC}
We first provide simple upper bounds to the SER of single-carrier-systems for $M$-PSK and $M$-DPSK modulation with IQI at the TX and/or RX.

\subsubsection{M-ary PSK}

It is recalled that the SER of single-carrier systems is given in (\ref{eq:SER_PSK}). It is evident that by setting $\theta= \pi / 2$, the SER is upper bounded by 
\begin{equation}
\label{eq:SC_bound}
P_{s,{\rm PSK}}\leq \tilde{M} e^{-\frac{g_{{\rm PSK}}\alpha}{\beta}+\frac{A}{\beta\overline{\gamma}}} \Gamma\left(1,\frac{A}{\overline{\gamma}\beta} ;-\frac{  g_{{\rm PSK}}\alpha A}{\beta^2\overline{\gamma}}\right )
\end{equation}
\noindent where $\tilde{M}={\left(M-1\right)}/{M}$, which for high SNR levels simplifies to
\begin{equation}
\label{eq:asym_PSK}
P_{s,{\rm PSK}}\leq \tilde{M} e^{-\frac{g_{{\rm PSK}}\alpha}{\beta}}. 
\end{equation}

This upper bound provides insights into the asymptotic behavior of the considered system. For instance, assuming TX or RX IQI only, $\frac{\alpha}{\beta}=\rm{IRR}_{t/r}$; hence, as $\rm{IRR}_{t/r}$ approaches $\infty$, $P_{s,PSK}\rightarrow 0$. On the contrary, as $\rm{IRR}_{t/r}$ approaches $1$, $P_{s,PSK}\rightarrow\frac{M-1}{M}e^{-g_{PSK}}$ which is directly proportional to $M$. Hence, a higher modulation order implies a higher error floor. Moreover, it is evident that the asymptotic behavior of the SER depends on both the modulation index and the IQI parameters.

\subsubsection{M-ary DPSK}

Assuming noncoherent $M$-ary DPSK, the SER is obtained by substituting (\ref{eq:SC_mgf}) in (\ref{eq:dpsk_mgf}). Hence, by setting $\theta=0$, the SER can be upper bounded as follows 
\begin{equation}
\label{eq:SC_bound_DPSK}
P_{s,{\rm DPSK}}\leq \tilde{M} e^{-\frac{\alpha}{\beta}\frac{g_{{\rm PSK}}}{1+\rho}+\frac{A}{\beta\overline{\gamma}}}\Gamma\left(1,\frac{A}{\overline{\gamma}\beta} ;\frac{-g_{{\rm PSK}} \alpha A}{\left(1+\rho \right)\beta^2\overline{\gamma}}\right )   
\end{equation}

\noindent which for high SNR values, since $\Gamma\left(1,0,0\right )=1$, simplifies to
\begin{equation}
\label{eq:asym_DPSK}
P_{s,{\rm DPSK}}\leq \tilde{M} e^{-\frac{\alpha}{\beta}\frac{g_{\rm PSK}}{1+\rho}}.
\end{equation}
\noindent We observe that the exponential function argument in (\ref{eq:asym_DPSK}) is similar to the argument in (\ref{eq:asym_PSK}) but divided by $1+\rho>1$. Hence, from the derived upper bound, we can conclude that for a fixed $M$, the SER of DPSK is asymptotically greater than the SER of PSK.  

\subsection {Multi-Carrier Systems}

%For multi-carrier systems with joint TX/RX IQI the SER of $M$-PSK, $M$-DPSK and $M$-FSK can be obtained by substituting (\ref{eq:MC_JOINT_MGF}) in (\ref{eq:PSK_mgf}), (\ref{eq:dpsk_mgf}) and (\ref{eq:PSK_mgf}), respectively. 
In this subsection, upper bounds and asymptotic expressions are derived for the SER of the considered modulation schemes for multi-carrier systems with joint TX and/or RX IQI. %These expressions can be simplified to the cases of TX IQI only and RX IQI only as shown in Section \ref{S:MC}. 

\subsubsection{M-ary PSK}

Assuming coherent $M$-ary PSK, the SER of multi-carrier systems with TX IQI only, RX IQI only and joint TX/RX IQI is obtained by substituting (\ref{eq:MC_mgf}), (\ref{eq:MC_RX_MGF}) and (\ref{eq:MGF_MC_joint_eq})$-$(\ref{eq:MGF_MC_joint_gt}) in (\ref{eq:PSK_mgf}), respectively.

\begin{itemize}
\item \textit{TX IQI and ideal RX: }
Based on the above and setting $\theta = \pi / 2$, one obtains
\begin{equation}
\label{eq:MC_TX_bound1}
P_{s,\rm{PSK}}\leq \tilde{M} e^{-\frac{g_{\rm{PSK}}|\mu_t|^{2}}{|\nu_t|^{2}}+\frac{1}{|\nu_t|^{2}\overline{\gamma}}} \Gamma\left(1,\frac{A}{\overline{\gamma}|\nu_t|^{2}} ;-\frac{  g_{\rm{PSK}}|\mu_t|^{2} }{|\nu_t|^{4}\overline{\gamma}}\right )
\end{equation}
\noindent which for high SNR values reduces to the following simple closed-form upper-bound
\begin{equation}
\label{eq:MC_TX_asym1}
P_{s,\rm{PSK}}\leq  \tilde{M} e^{-\frac{g_{\rm{PSK}}|\mu_t|^{2}}{|\nu_t|^{2}}}.
\end{equation}

It is noticed that (\ref{eq:MC_TX_bound1}) and (\ref{eq:MC_TX_asym1}) are similar to (\ref{eq:SC_bound}) and (\ref{eq:asym_PSK}) when  $\alpha=|\mu_t|^{2}$ and $\beta=|\nu_t|^{2}$. Importantly, this implies that under TX IQI only, single-carrier and multi-carrier systems exhibit similar behaviors.  

\item \textit{RX IQI and ideal TX: }
For RX IQI only, the SER is upper bounded by
\begin{equation}
P_{s,\rm{PSK}}\leq\tilde{M}\left(1+\frac{g_{\rm{PSK}}|\mu_r|^2}{|\nu_r|^2}e^{\frac{1}{\overline{\gamma}}+\frac{|\mu_r|^2}{\overline{\gamma}|\nu_r|^2}+\frac{g_{\rm{PSK}}|\mu_r|^2}{|\nu_r|^2}}\textup{Ei}\left(- \frac{1}{\overline{\gamma}}-\frac{|\mu_r|^2}{\overline{\gamma}|\nu_r|^2}-\frac{g_{\rm{PSK}}|\mu_r|^2}{|\nu_r|^2}\right ) \right).
\end{equation}
\noindent It is evident that for asymptotic SNR values, the above inequality simplifies to
\begin{equation}
P_{s,\rm{PSK}}\leq\tilde{M}\left(1+\frac{g_{\rm{PSK}}|\mu_r|^2}{|\nu_r|^2}e^{\frac{g_{\rm{PSK}}|\mu_r|^2}{|\nu_r|^2}}\textup{Ei}\left(-\frac{g_{\rm{PSK}}|\mu_r|^2}{|\nu_r|^2}\right ) \right).
\end{equation}

Also, as $\rm{IRR}_{r}=|\mu_r|^2 / |\nu_r|^2$ approaches $\infty$, the exponential integral function can be approximated by $\textup{Ei}\left(-z\right)\approx-\frac{e^{-z}}{z}\left(1-\frac{1}{z}+\frac{2!}{z^2}-\dots\right)$ \cite{abramowitz1966handbook} and hence $P_{s, \rm PSK}\rightarrow 0$. Likewise, as $\rm{IRR}_{r}$ approaches unity, one obtains 
\begin{equation}
 P_{s,\rm PSK}\rightarrow P_{s,\rm PSK}\leq\tilde{M}\left(1+g_{\rm PSK}e^{g_{\rm PSK}}\textup{Ei}\left(-g_{\rm PSK}\right ) \right).
\end{equation}
\noindent Moreover, it is noted that $\forall x\geq 0$ and $y=1+xe^x\textup{Ei}\left(x \right )$, we have $x \propto 1/y$. As a result, it follows that $P_{s,\rm PSK} \propto 1/\rm{IRR}_r$ and $P_{s,\rm PSK} \propto M$. 

\item \textit{Joint TX/RX IQI: }
Finally, for joint TX/RX IQI with $|\xi_{12}|^2=|\xi_{21}|^2$ i.e. ${\rm IRR}_t={\rm IRR}_r$ the corresponding SER is upper bounded by
\begin{equation}
%\begin{split}
P_{s,\rm{PSK}}\leq\tilde{M}\Bigg(C  -\frac{|\xi_{12}|^2e^{-g_{\rm{PSK}}\frac{|\xi_{11}|^2}{|\xi_{12}|^2}+\frac{|\mu_{r}|^{2}+|\nu_{r}|^{2}}{|\xi_{12}|^2 \overline{\gamma}}}}{g_{\rm{PSK}}\left(|\xi_{11}|^2-|\xi_{22}|^2 \right )} \gamma\left(2,-g_{\rm{PSK}}\frac{|\xi_{11}|^2}{|\xi_{12}|^2};-g_{\rm{PSK}}\frac{|\xi_{11}|^2\left( |\mu_{r}|^{2}+|\nu_{r}|^{2}\right )}{|\xi_{12}|^4\overline{\gamma}} \right)  \Bigg )
%\end{split}
\end {equation}
\noindent which for asymptotic SNR values simplifies to
\begin{equation}
P_{s,\rm{PSK}}\leq\tilde{M}\left(C-\frac{|\xi_{12}|^2e^{-g_{\rm{PSK}}\frac{|\xi_{11}|^2}{|\xi_{12}|^2}}}{g_{\rm{PSK}}\left(|\xi_{11}|^2-|\xi_{22}|^2 \right )}\gamma\left(2,-g_{\rm{PSK}}\frac{|\xi_{11}|^2}{|\xi_{12}|^2}\right)  \right )
\end{equation}
\noindent where $\gamma\left(a,x\right)=\int_{0}^{x}t^{a-1}e^{-t}{\rm d}t$ is the lower incomplete gamma function \cite{B:Tables}. It is noted that $\forall x\geq 0$ and $y=1-\frac{e^{-x}}{x}\gamma\left(2,-x \right )$, we have $x \propto 1/y$ and thus, $P_{s, \rm PSK}\propto M$. Moreover, since $|\xi_{11}|^2/|\xi_{12}|^2=\rm{IRR}_t=\rm{IRR}_r$, it follows that $P_{s,\rm PSK} \propto 1/\rm{IRR}_{t/r}$.
\end{itemize}

\subsubsection{M-ary DPSK}

Assuming noncoherent $M$-ary DPSK, the SER of multi-carrier systems with TX IQI only, RX IQI only and joint TX/RX IQI is obtained by substituting (\ref{eq:MC_mgf}), (\ref{eq:MC_RX_MGF}) and (\ref{eq:MGF_MC_joint_eq})$-$(\ref{eq:MGF_MC_joint_gt}) in (\ref{eq:dpsk_mgf}), respectively.  
\begin{itemize}
\item \textit{TX IQI and ideal RX:} 
Based on the above and by setting $\theta = 0$, it follows that 
\begin{equation}
\label{eq:MC_TX_bound}
P_{s,{\rm DPSK}}\leq \tilde{M}e^{-\frac{g_{{\rm PSK}}|\mu_t|^{2}}{\left(1+\rho\right)|\nu_t|^{2}}+\frac{1}{|\nu_t|^{2}\overline{\gamma}}} \Gamma\left(1,\frac{A}{\overline{\gamma}|\nu_t|^{2}} ;-\frac{  g_{{\rm PSK}}|\mu_t|^{2} }{\left(1+\rho\right)|\nu_t|^{4}\overline{\gamma}}\right )\end{equation}
\noindent which for high SNR values reduces to the following simple bound
\begin{equation}
\label{eq:MC_TX_asym}
P_{s,{\rm DPSK}}\leq  \tilde{M} e^{-\frac{g_{{\rm PSK}}|\mu_t|^{2}}{\left(1+\rho\right)|\nu_t|^{2}}}.
\end{equation}
\noindent It is noted that (\ref{eq:MC_TX_bound}) and (\ref{eq:MC_TX_asym}) are similar to (\ref{eq:SC_bound_DPSK}) and (\ref{eq:asym_DPSK}) when $\alpha=|\mu_t|^{2}$ and $\beta=|\nu_t|^{2}$. Therefore, in the case of TX IQI only, $M$-DPSK based single-carrier and multi-carrier systems show similar behaviors.

\item \textit{RX IQI and ideal TX:}
For RX IQI only, the SER is upper bounded by
\begin{equation}
%\begin{split}
P_{s,{\rm DPSK}}\leq\tilde{M}\Bigg(1 +\frac{g_{{\rm PSK}}|\mu_r|^2}{\left(1+\rho\right)|\nu_r|^2}e^{\frac{1}{\overline{\gamma}}+\frac{|\mu_r|^2}{\overline{\gamma}|\nu_r|^2}+\frac{g_{{\rm PSK}}|\mu_r|^2}{\left(1+\rho\right)|\nu_r|^2}}   \textup{Ei}\left(- \frac{1}{\overline{\gamma}}-\frac{|\mu_r|^2}{\overline{\gamma}|\nu_r|^2}-\frac{g_{P{\rm PSK}}|\mu_r|^2}{\left(1+\rho\right)|\nu_r|^2}\right ) \Bigg)
%\end{split}
\end{equation}
\noindent which for asymptotic SNR values simplifies to
\begin{equation}
P_{s,{\rm DPSK}}\leq\tilde{M}\left(1+\frac{g_{{\rm PSK}}|\mu_r|^2}{\left(1+\rho\right)|\nu_r|^2}e^{\frac{g_{{\rm PSK}}|\mu_r|^2}{\left(1+\rho\right)|\nu_r|^2}}\textup{Ei}\left(-\frac{g_{{\rm PSK}}|\mu_r|^2}{\left(1+\rho\right)|\nu_r|^2}\right ) \right).
\end{equation}
\noindent Notably, since $\frac{g_{PSK}|\mu_r|^2}{|\nu_r|^2}>\frac{g_{PSK}|\mu_r|^2}{\left(1+\rho\right)|\nu_r|^2}$, we can conclude that for a fixed $M$, the SER of DPSK is asymptotically greater than the SER of PSK. 

\item \textit{Joint TX/RX IQI:}
Finally, for the case of joint TX/RX IQI with $|\xi_{12}|^2=|\xi_{21}|^2$, the SER is upper bounded by
\begin{equation}
\begin{split}
P_{s,{\rm DPSK}}\leq\tilde{M}\Bigg(C & -\frac{\left(1+\rho \right )|\xi_{12}|^2e^{-\frac{g_{{\rm PSK}}}{1+\rho}\frac{|\xi_{11}|^2}{|\xi_{12}|^2}+\frac{|\mu_{r}|^{2}+|\nu_{r}|^{2}}{|\xi_{12}|^2 \overline{\gamma}}}}{g_{{\rm PSK}}\left(|\xi_{11}|^2-|\xi_{22}|^2 \right )} \\ & \times \gamma\left(2,-\frac{g_{{\rm PSK}}}{\left(  1+\rho\right)}\frac{|\xi_{11}|^2}{|\xi_{12}|^2};-\frac{g_{{\rm PSK}}}{\left( 1+\rho\right )}\frac{|\xi_{11}|^2\left( |\mu_{r}|^{2}+|\nu_{r}|^{2}\right )}{|\xi_{12}|^4\overline{\gamma}} \right)  \Bigg )
\end{split}
\end {equation}
\noindent which for asymptotic SNR values simplifies to
\begin{equation}
P_{s,{\rm DPSK}}\leq\tilde{M}\left(C-\frac{\left(1+\rho \right )|\xi_{12}|^2e^{-\frac{g_{{\rm PSK}}|\xi_{11}|^2}{\left(1+\rho \right )|\xi_{12}|^2}}}{g_{{\rm PSK}}\left(|\xi_{11}|^2-|\xi_{22}|^2 \right )}\gamma\left(2,-\frac{g_{{\rm PSK}}|\xi_{11}|^2}{\left(1+\rho \right )|\xi_{12}|^2}\right)  \right ).  
\end{equation}
It is also shown, for this case, that for a fixed $M$, the SER of DPSK is asymptotically greater than the SER of its PSK counterpart. 
\end{itemize}
\begin{figure}[h]
\centering
\includegraphics[width=12.0cm, height = 10.0cm]{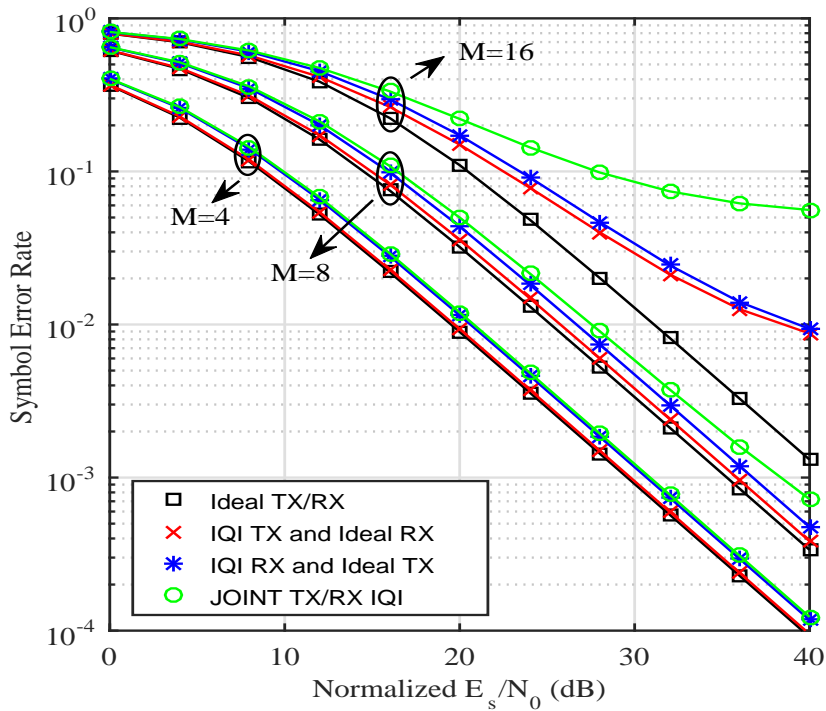}
\caption{Single-carrier system average SER as a function of the normalized $E_s/N_0$ for \textit{M}-PSK when ${\rm IRR}_t={\rm IRR}_r=20 {\rm dB}$ and $\phi=3\degree$.}
\label{fig:SC_PSK}
\end{figure}
\section {Numerical and Simulation Results}
\label{S:Sim}

In this section, we quantify the effects of IQI on the performance of single-carrier and multi-carrier based \textit{M}-PSK, \textit{M}-DPSK and \textit{M}-FSK systems over flat Rayleigh fading channels in terms of the corresponding average SER. For a fair comparison, we assume that the transmit power level is always fixed. This implies that the transmitted signal is normalized by $|\mu_{t}|^{2}+|\nu_{t}|^{2}$ for TX IQI, by $|\mu_{r}|^{2}+|\nu_{r}|^{2}$ for RX IQI and by $\left(|\mu_{t}|^{2}+|\nu_{t}|^{2}\right)\left(|\mu_{r}|^{2}+|\nu_{r}|^{2}\right)$ for joint TX/RX IQI. 

\begin{figure}[h!]
\centering
\includegraphics[width=12.0cm, height = 10.0cm]{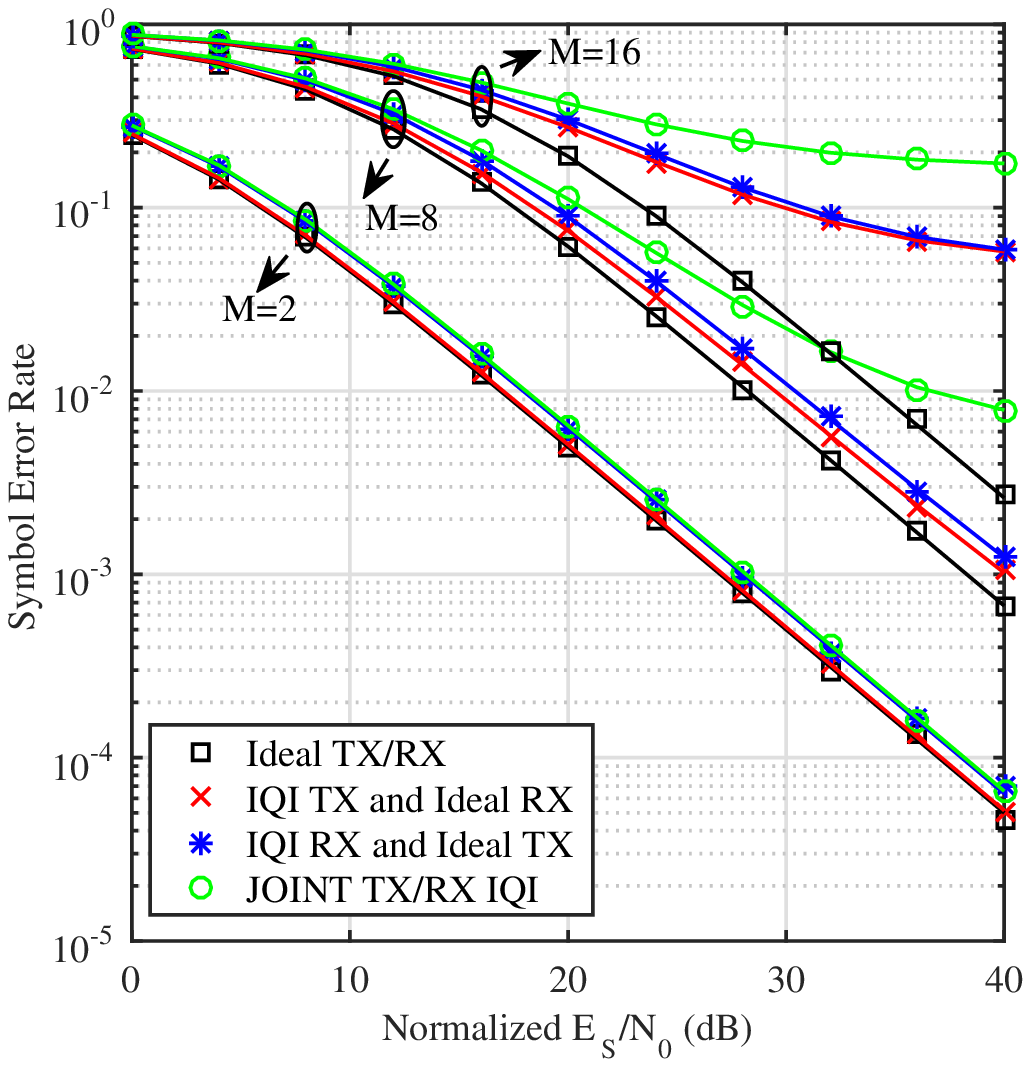}
\caption{Single-carrier system average SER as a function of the normalized $E_s/N_0$ for \textit{M}-DPSK when ${\rm IRR}_t={\rm IRR}_r=20 {\rm dB}$ and $\phi=3\degree$.  }
\label{fig:SC_DPSK}
\end{figure}
\begin{figure}[h!]
\centering
\includegraphics[width=12.0cm, height = 10.0cm]{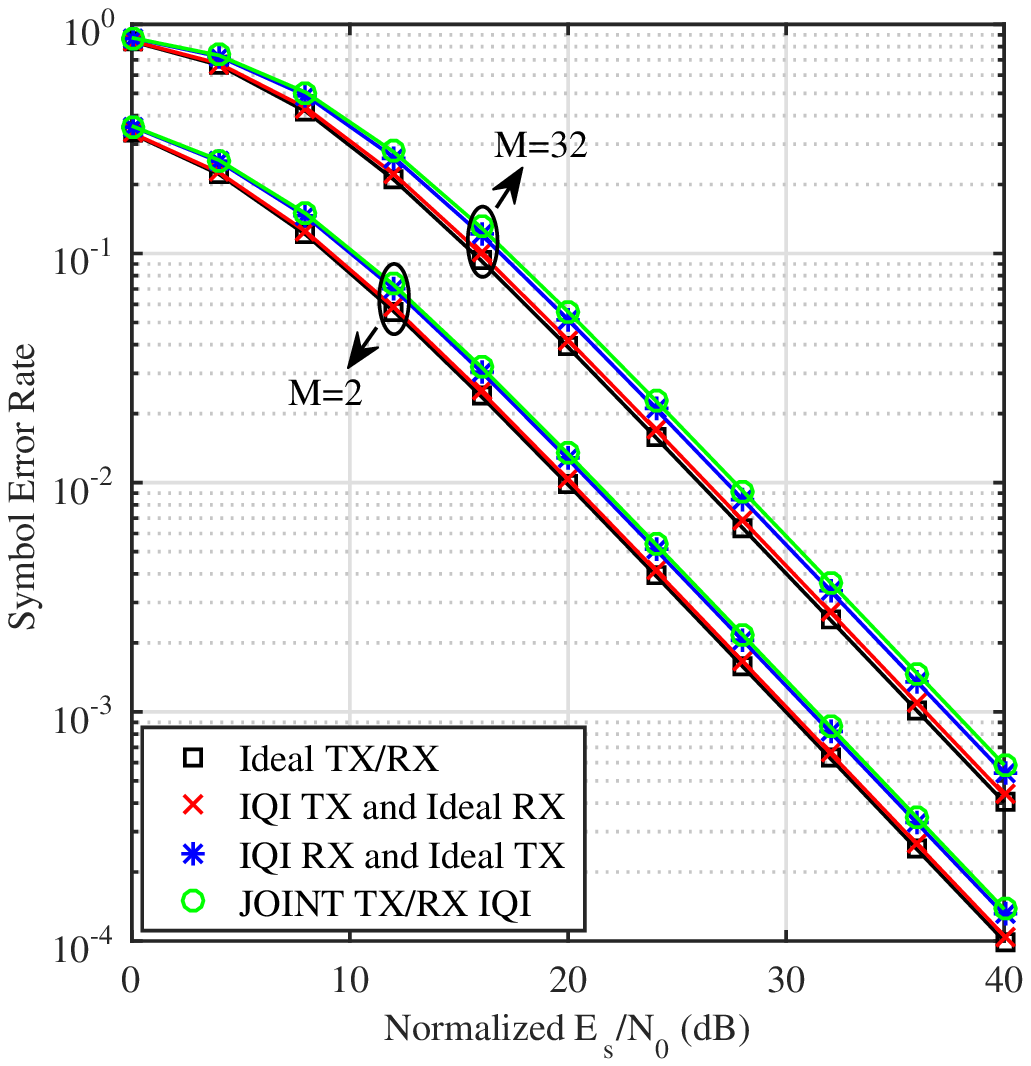}
\caption{Single-carrier system average SER as a function of the normalized $E_s/N_0$ for \textit{M}-FSK when ${\rm IRR}_t={\rm IRR}_r=20 {\rm dB}$ and $\phi=3\degree$.}
\label{fig:SC_FSK}
\end{figure}
\begin{figure}[h]
\centering
\includegraphics[width=12.0cm, height = 10.0cm]{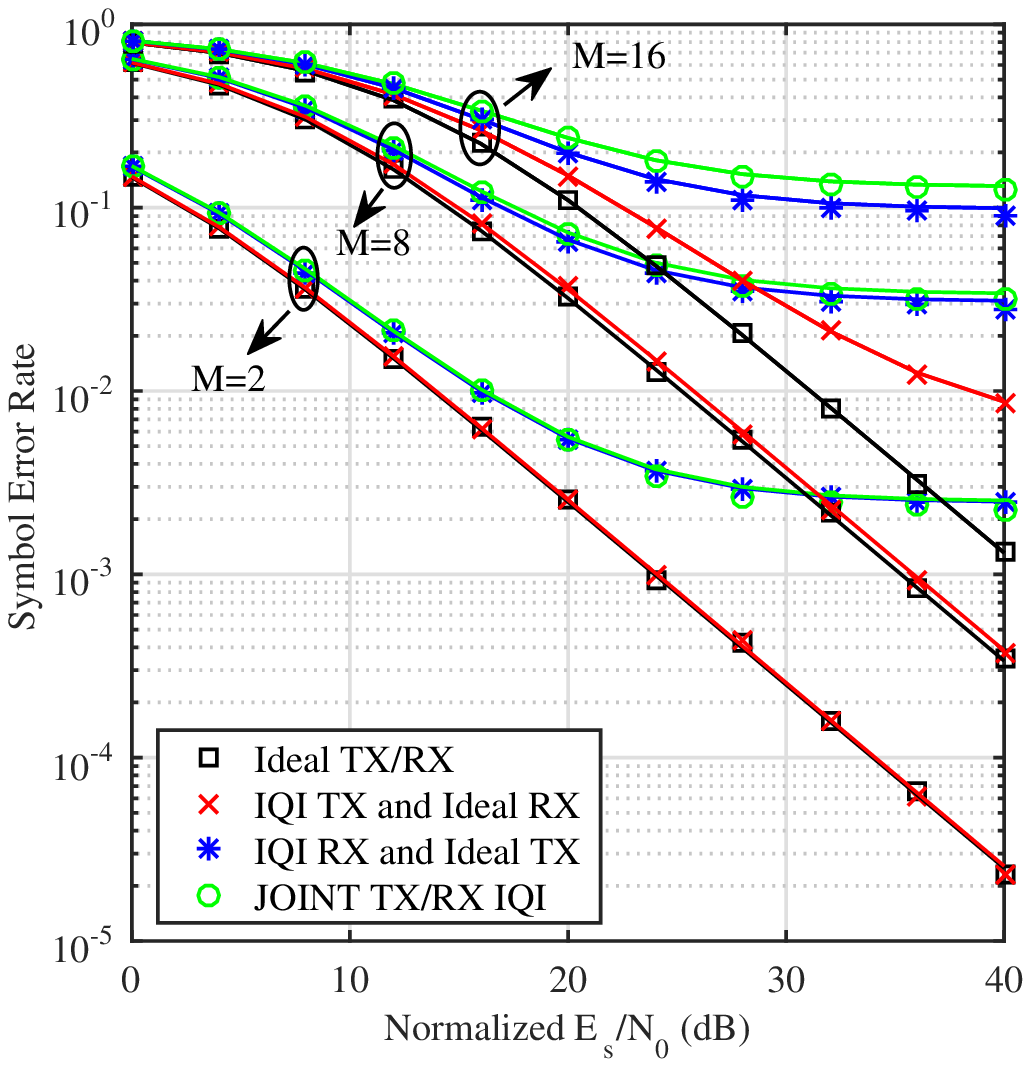}
\caption{Multi-carrier system average SER as a function of the normalized $E_s/N_0$ for \textit{M}-PSK when ${\rm IRR}_t={\rm IRR}_r=20 {\rm dB}$ and $\phi=3\degree$.}
\label{fig:MC_PSK}
\end{figure}
To this end, Figs. \ref{fig:SC_PSK}$-$\ref{fig:SC_FSK} and Figs. \ref{fig:MC_PSK}$-$\ref{fig:MC_FSK} illustrate the SER for \textit{M}-PSK, \textit{M}-DPSK and \textit{M}-FSK for single-carrier systems and multi-carrier systems, respectively. Assuming ${\rm IRR}_t={\rm IRR}_r=20{\rm dB}$, all possible combinations of ideal/impaired TX/RX are presented. It is noted that the numerical results are shown with continuous lines, whereas markers are used to illustrate the respective computer simulation results. For both single-carrier and multi-carrier\footnote{This demonstrates that our assumption of uncorrelated carrier and its image does not affect the accuracy of the SER analysis in multi-carrier systems.} systems, it is noticed that the derived expressions characterize accurately the simulated SER performance for all the considered modulation schemes in the presence of IQI. Specifically, it is first observed that RX IQI has more detrimental impact on the system performance than TX IQI. This result is expected since RX IQI affects both the signal and the noise while TX IQI impairs the information signal only. Second, it is noticed that IQI exhibits different effects on the different modulation schemes considered. For example, it can be drawn from Fig.\ref{fig:SC_FSK} that the effects of IQI on single-carrier FSK are rather limited irrespective of the modulation order. This can be explained by the fact that the tone spacing in FSK is constant regardless of the modulation order. Hence, unlike PSK and DPSK, the IQI effects on FSK do not depend on the modulation order for both single-carrier and multi-carrier systems. However, the cost of increasing $M$ for FSK is an increased transmission bandwidth. This is not the case for the other two modulation schemes where the angle separation depends on the modulation order. For instance, the effects of IQI can be considered acceptable i.e., no error floor observed for the considered SNR range, only for $M\leq8$ and $M\leq 4$ for the cases of PSK and DPSK modulations, respectively. In fact, for single-carrier systems, when $M=16$, an error floor is observed at around $30{\rm dB}$ when PSK modulation suffers from joint TX/RX IQI, while for DPSK this error floor appears at around $28{\rm dB}$ for all the considered impairment scenarios. It is also worth mentioning that for the joint TX/RX IQI case, this error floor is around $6 \times 10^{-2}$ for PSK versus $2 \times 10^{-1}$ for DPSK. Hence for a fixed $M$, the error floor is higher for DPSK than PSK, which confirms our observations in Section \ref{S:Asymptotic}.
\begin{figure}[h!]
\centering
\includegraphics[width=12.0cm, height = 10.0cm]{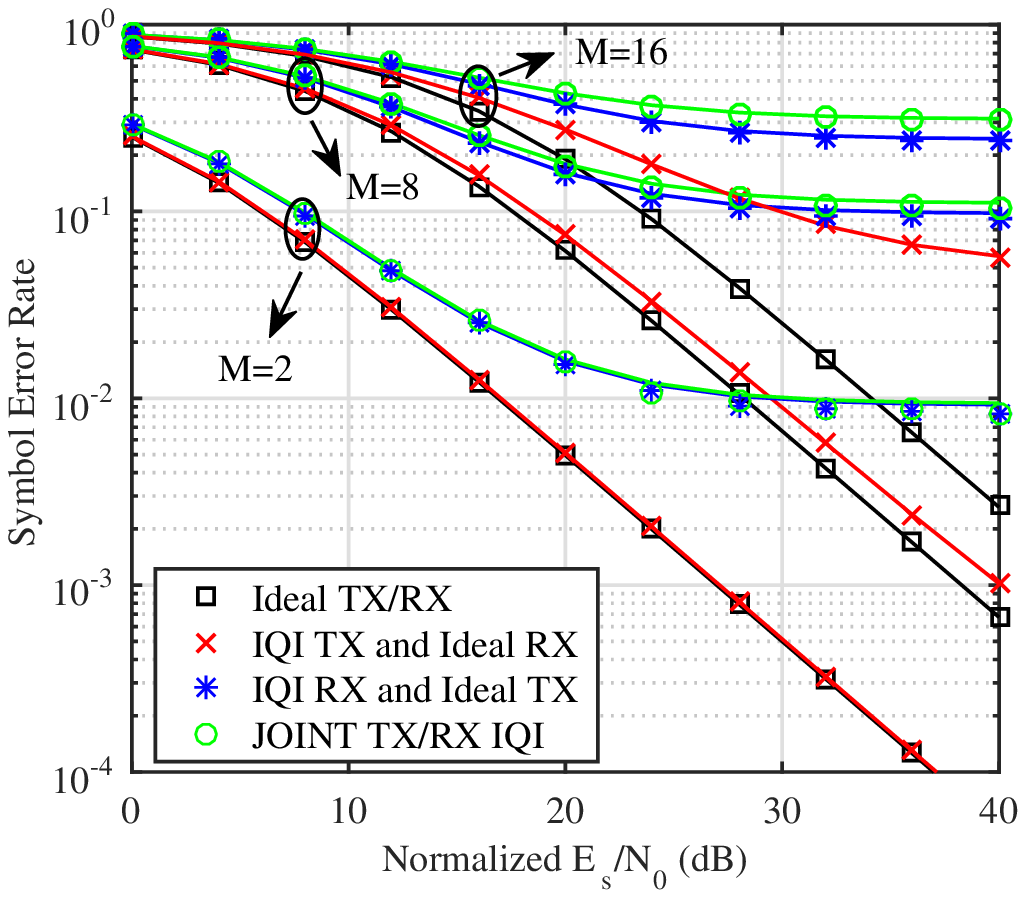}
\caption{Multi-carrier system average SER as a function of the normalized $E_s/N_0$ for \textit{M}-DPSK when ${\rm IRR}_t={\rm IRR}_r=20 {\rm dB}$ and $\phi=3\degree$.}
\label{fig:MC_DPSK}
\end{figure}
\begin{figure}[h!]
\centering
\includegraphics[width=12.0cm, height = 10.0cm]{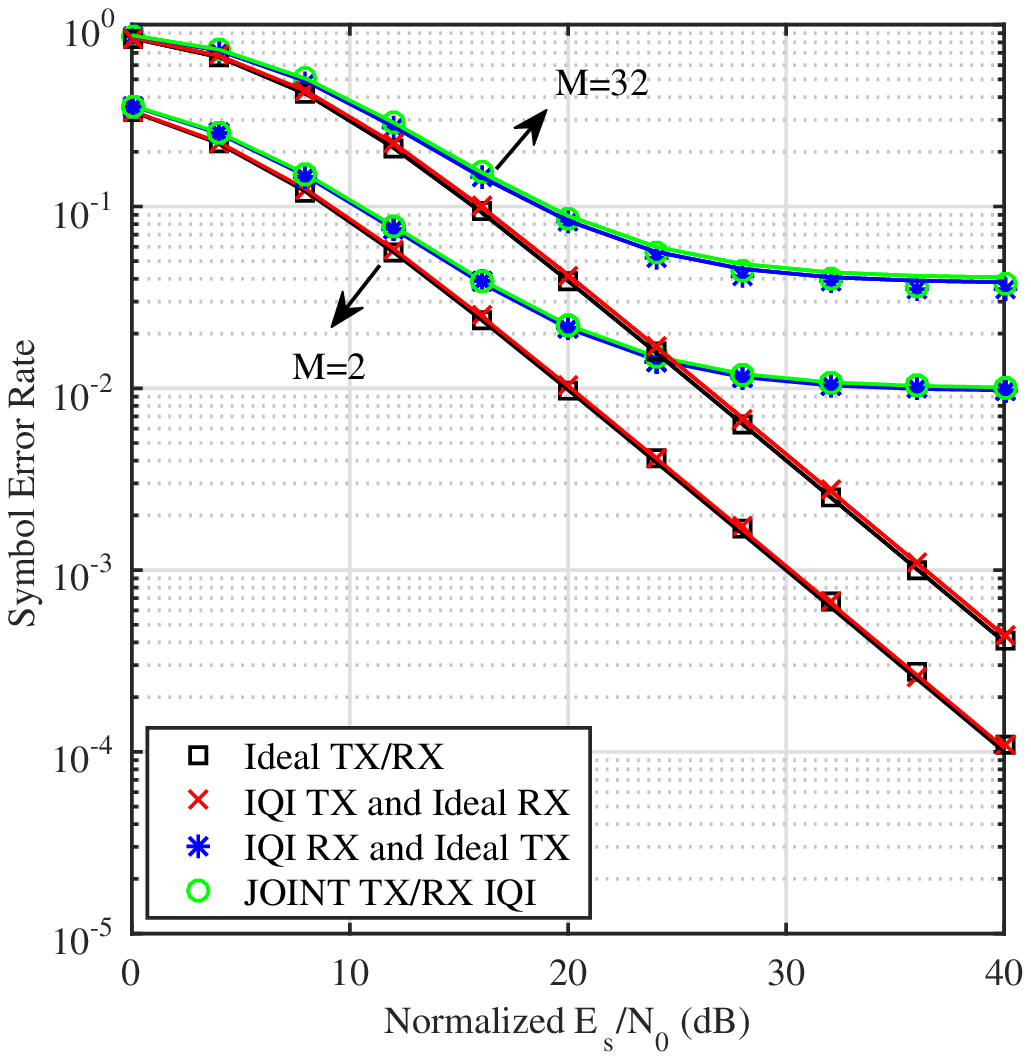}
\caption{Multi-carrier system average SER as a function of the normalized $E_s/N_0$ for \textit{M}-FSK when ${\rm IRR}_t={\rm IRR}_r=20 {\rm dB}$ and $\phi=3\degree$.}
\label{fig:MC_FSK}
\end{figure}
\begin{figure}[h!]
\centering
\includegraphics[width=12.0cm, height = 10.0cm]{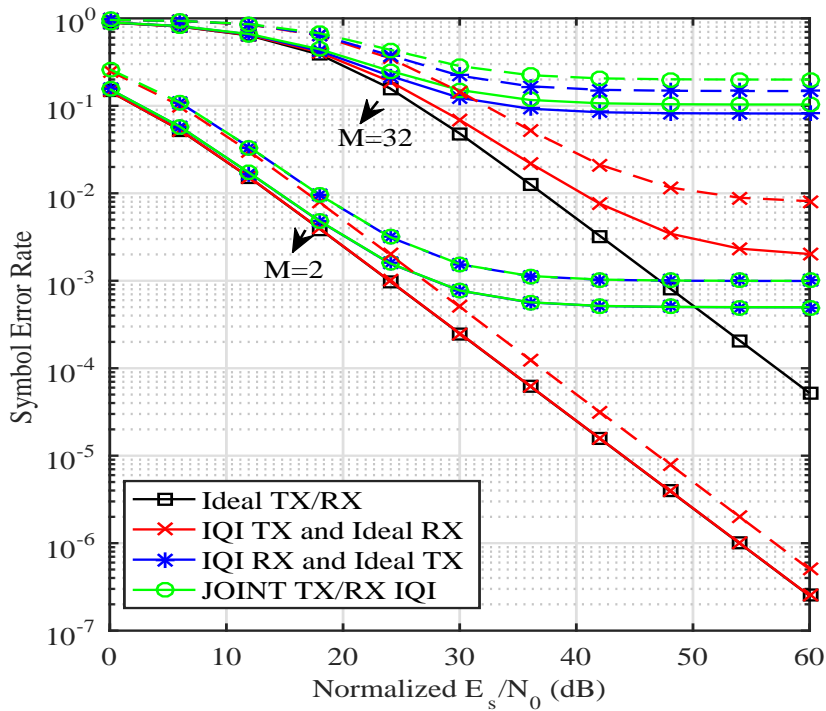}
\caption{Multi-carrier system average SER (solid line) and derived upper bound (dashed line) as a function of the normalized $E_s/N_0$ for \textit{M}-PSK when ${\rm IRR}_t={\rm IRR}_r=27 {\rm dB}$ and $\phi=1\degree$.}
\label{fig:MC_bound}
\end{figure}
\begin{figure}[h]
\centering
\includegraphics[width=12.0cm, height = 10.0cm]{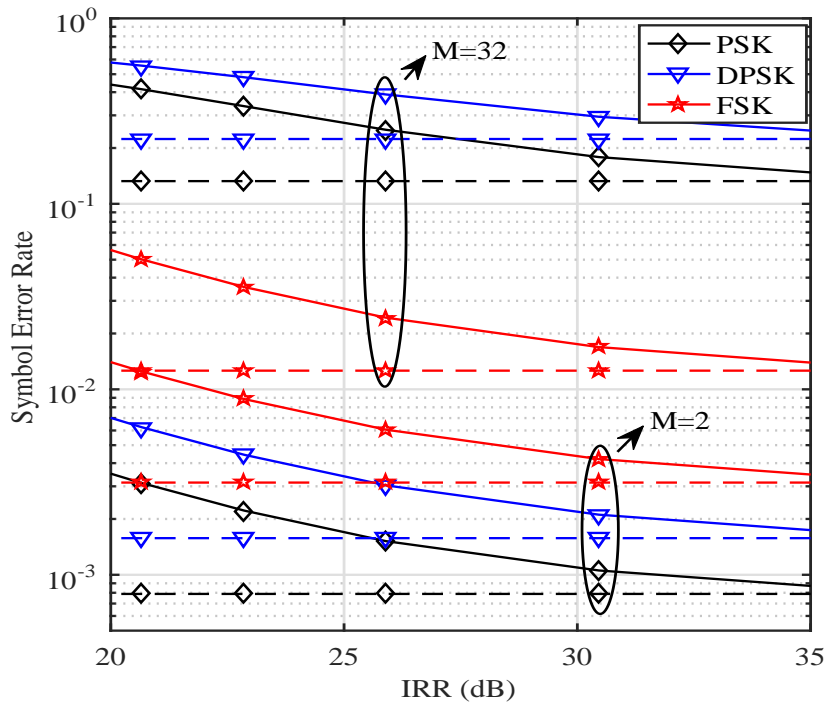}
\caption{Multi-carrier system average SER as a function of the $\rm IRR$ for \textit{M}-PSK, \textit{M}-DPSK and \textit{M}-FSK, with RX IQI only, when $E_s/N_0=25 {\rm dB}$ and $\phi=2\degree$.}
\label{fig:IRR}
\end{figure}
\begin{figure}[h]
\centering
\includegraphics[width=12.0cm, height = 10.0cm]{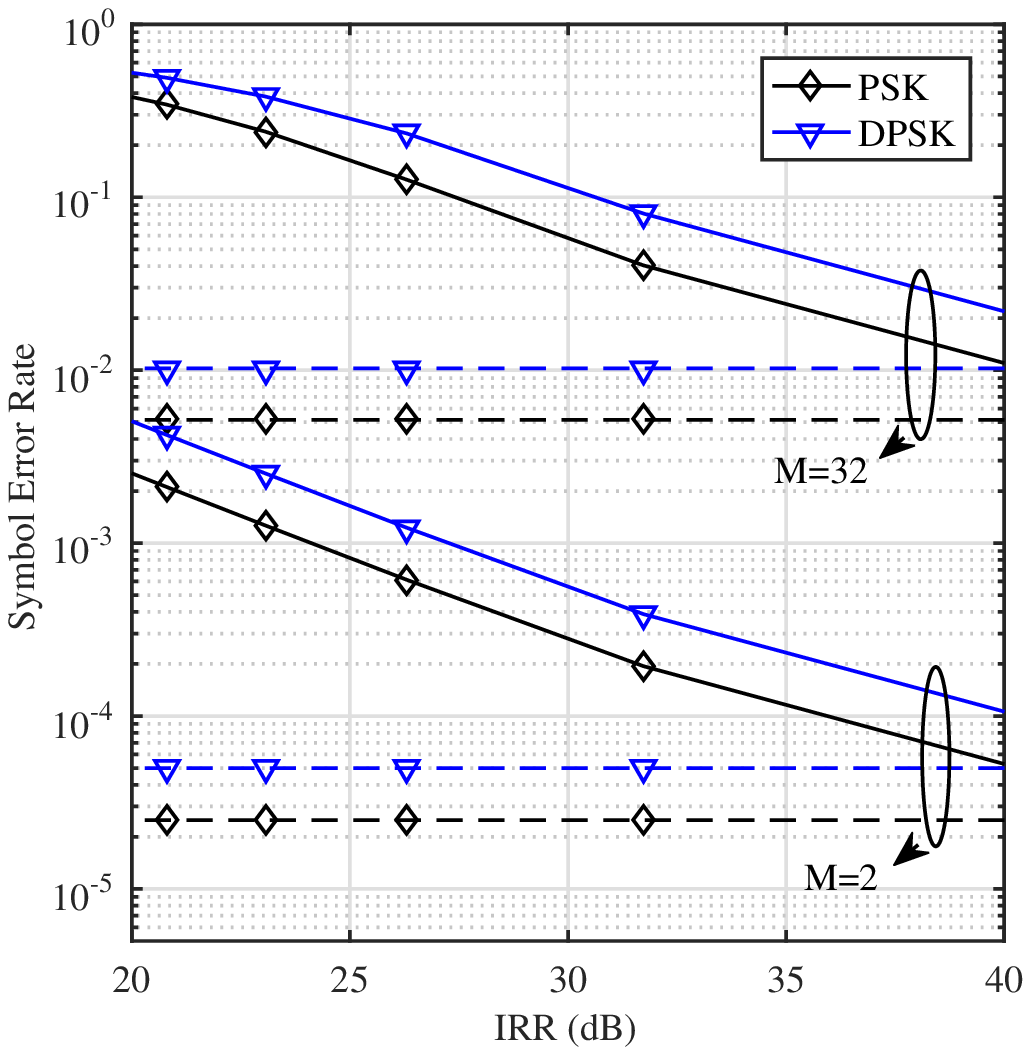}
\caption{Multi-carrier system average SER as a function of the $\rm IRR$ for \textit{M}-PSK and \textit{M}-DPSK, with RX IQI only, when $E_s/N_0=40 {\rm dB}$ and $\phi=1\degree$.}
\label{fig:IRR_40}
\end{figure}

Even though the effects of IQI on the different modulation schemes follow the same trend in multi-carrier systems as in single-carrier systems, it is observed that IQI affects the former more severely than the latter. This is because IQI in multi-carrier systems causes interference from the image subcarrier, which can have higher SNR than the desired signal, while single-carrier IQI causes interference from the signal's own complex conjugate. An interesting example is the case of \textit{M}-FSK constellation, where in single-carrier systems the effects of IQI are negligible, while in multi-carrier systems an error floor is observed in Fig. \ref{fig:MC_FSK}, regardless of the modulation order, for the RX IQI only as well as for joint TX/RX IQI cases. In the same context, the error floor for binary FSK appears at around $24 {\rm dB}$. This error floor is observed for PSK and DPSK as well with binary PSK being the most robust to IQI among the considered modulations, since the error floor appears at around $28{\rm dB}$. It is also noted that unlike single-carrier systems where in some cases IQI could be neglected, for the considered $\rm{IRR}$ values, in multi-carrier systems the effects of IQI at the RX should be compensated in order to achieve a reliable communication link, even in the case of the relatively simple binary modulation schemes.

For multi-carrier systems, Fig. \ref{fig:MC_bound} compares the derived upper bound to the exact SER of $M$-PSK when ${\rm IRR}_t={\rm IRR}_r=27{\rm dB}$. In this scenario, we consider the cases of TX IQI only, RX IQI only and joint TX/RX IQI for $M=2$ and $M=32$. The solid lines correspond to the exact SER while the dashed lines represent the respective bound. It is noticed that although the bound is not particularly tight, it exhibits the same behavior as the exact SER curves and hence can provide useful insights into the system performance. 

Finally, Fig. \ref{fig:IRR} and Fig. \ref{fig:IRR_40} demonstrate the effects of the $\rm IRR$ on the SER of the different considered modulation schemes for multi-carrier systems, when ${\rm SNR}=25{\rm dB}$ and ${\rm SNR}=40{\rm dB}$, respectively. It is assumed that both TX and RX are IQI-impaired and that ${\rm IRR}_t={\rm IRR}_r$. The phase imbalance assumed is $1\degree$ in Fig. \ref{fig:IRR} and $2\degree$ in Fig. \ref{fig:IRR_40}. It is also noted that the continuous lines and dashed lines correspond to the IQI-impaired and ideal cases, respectively. For moderate SNR values, one can see that IQI affects the different modulations schemes in a different manner. For instance, joint TX/ RX IQI exhibits a constant loss in the SER performance of $M$-FSK regardless of the modulation order, which is not the case when considering phase modulation. Moreover, it is noticed that for lower SNR values, the effects of IQI vanish when the $\rm IRR$ is increased; however, for higher SNR values and given that IQI effects dominate noise effects at high SNR, there is a noticeable performance degradation even when considering high $\rm IRR$ values.

\section {Conclusion}
\label{S:Conclusions}
We developed a general framework for the SER performance analysis of different \textit{M}-ary coherent and non-coherent modulation schemes over Rayleigh fading channels in the presence of IQI at the RF front end. The realistic cases of TX IQI only, RX IQI only and joint TX/RX IQI were considered and the corresponding average SER expression of the underlying schemes was derived both in exact and in asymptotic form providing useful insights into the overall system behavior. The derived analytic results were corroborated with respective results from computer simulations. It was shown that the performance degradation caused by IQI depends on the considered modulation scheme with \textit{M}-DPSK being the most sensitive modulation scheme to IQI. Moreover, for coherent and noncoherent phase modulation, increasing the modulation order increases the impact of IQI on the system, while for frequency modulation the performance degradation observed is constant regardless of the modulation order and single carrier frequency modulation is the most robust scheme to IQI effects.

\appendices
\section{Derivation of MGF for Multi-Carrier Systems Impaired by Joint TX/RX IQI}
\label{Appendix:MGF_MC_JOINT}

From (\ref{eq:def_mgf}) and (\ref{eq:MC_JOINT_PDF}), taking $u=e^{s \gamma}$ and $dv=f_{\gamma}\left(\gamma \right )$ and integrating by parts, one obtains
\begin{equation}
\label{eq:MC_JOINT_MGF_INT}
\mathcal{M}_{\gamma_{\rm IQI}}\left(s\right)=C+s\int_{0}^{\frac{|\xi_{11}|^2}{|\xi_{12}|^2}}\frac{ |\xi_{11}|^{2}-\gamma|\xi_{12}|^{2}}{|\xi_{11}|^2-|\xi_{22}|^2+x\left(|\xi_{21}|^2 -|\xi_{12}|^2\right )}e^{s x }e^{-\frac{x}{\overline{\gamma}}\left( \frac{|\mu_{R}|^{2} +\left|\nu_{R}\right|^{2}}{|\xi_{11}|^{2}-x|\xi_{12}|^{2}}\right )} {\rm d}x 
\end{equation}
\noindent where $C$ is given in (\ref{eq:const}). For the case of $|\xi_{12}|^2=|\xi_{21}|^2$ and setting $z=|\xi_{11}|^{2}-x|\xi_{12}|^{2}$, equation (\ref{eq:MC_JOINT_MGF_INT}) simplifies to 
\begin{equation}
\mathcal{M_{\gamma_{\rm IQI}}}\left(s\right)=C+\frac{s}{|\xi_{12}|^2\left(|\xi_{11}|^2-|\xi_{22}|^2 \right )}e^{s\frac{|\xi_{11}|^2}{|\xi_{12}|^2}+\frac{|\mu_{r}|^{2}+|\nu_{r}|^{2}}{|\xi_{12}|^2 \overline{\gamma}}}\int_{0}^{|\xi_{11}|^{2}}z e^{-z \frac{s}{|\xi_{12}|^2}-\frac{|\xi_{11}|^2\left( |\mu_{r}|^{2}+|\nu_{r}|^{2}\right )}{\overline{\gamma}|\xi_{12}|^2 z}}{\rm d}z
\end{equation}
\noindent which, considering the change of variable $y=\frac{z s }{|\xi_{12}|^2}$, is equivalent to (\ref{eq:MGF_MC_joint_eq}). On the contrary, for $|\xi_{12}|^2 \neq |\xi_{21}|^2$ and setting $z=|\xi_{11}|^{2}-xa|\xi_{12}|^{2}$, equation (\ref{eq:MC_JOINT_MGF_INT}) becomes
\begin{equation}
\label{eq:MGF_temp}
\mathcal{M}_{\gamma_{\rm IQI}}\left(s\right)=\frac{|\xi_{11}|^2}{|\xi_{11}|^2-|\xi_{22}|^2}+\frac{s \, e^{\frac{|\mu_{R}|^{2}+|\nu_{R}|^{2}}{|\xi_{12}|^2 \overline{\gamma}}+s \frac{|\xi_{11}|^2}{|\xi_{12}|^2}}}{|\xi_{12}|^2-|\xi_{21}|^2}\int_{0}^{|\xi_{11}|^2}\frac{z \, e^{-\frac{|\xi_{11}|^2\left( |\mu_{R}|^{2}+|\nu_{R}|^{2}\right )}{|\xi_{12}|^2 \overline{\gamma}z}-s \frac{z}{|\xi_{12}|^2}}}{\frac{d}{|\xi_{12}|^2-|\xi_{21}|^2}+z} {\rm d}z
\end{equation}
\noindent where $d$ is given in (\ref{eq:d}). For the case of $\left|\frac{{|\xi_{11}|^2|\xi_{21}|^2-|\xi_{22}|^2|\xi_{12}|^2}}{|\xi_{12}|^2-|\xi_{21}|^2}\right|<|\xi_{11}|^2$, we expand the involved binomial which yields 
\begin{equation}
\mathcal{M}_{\gamma_{\rm IQI}}\left(s\right)=\frac{|\xi_{11}|^2}{|\xi_{11}|^2-|\xi_{22}|^2}+\sum_{k=0}^{\infty }\frac{\left(-1\right)^ks \, d^k e^{\frac{|\mu_{R}|^{2}+|\nu_{R}|^{2}}{|\xi_{12}|^2 \overline{\gamma}}+s \frac{|\xi_{11}|^2}{|\xi_{12}|^2}}}{\left( |\xi_{12}|^2-|\xi_{21}|^2\right )^{k+1}}\int_{0}^{|\xi_{11}|^2}{ z^{-k} \, e^{-\frac{|\xi_{11}|^2\left( |\mu_{R}|^{2}+|\nu_{R}|^{2}\right )}{|\xi_{12}|^2 \overline{\gamma}z}-s \frac{z}{|\xi_{12}|^2}}}{\rm d}z
\end{equation}
\noindent By setting once more $y= x s / |\xi_{12}|^2$, equation (\ref{eq:MGF_MC_joint_lt})) is deduced. Meanwhile for $\left|\frac{{|\xi_{11}|^2|\xi_{21}|^2-|\xi_{22}|^2|\xi_{12}|^2}}{|\xi_{12}|^2-|\xi_{21}|^2}\right| >  |\xi_{11}|^2$, and expanding the binomial in (\ref{eq:MGF_temp}), one obtains the following analytic expression
%
%\frac{|\xi_{11}|^2\left( |\mu_{R}|^{2}+|\nu_{R}|^{2}\right )}{|\xi_{12}|^2 \overline{\gamma}x}$ yields (\ref{eq:MGF_MC_joint_lt}). Meanwhile for $\left|\frac{{|\xi_{11}|^2|\xi_{21}|^2-|\xi_{22}|^2|\xi_{12}|^2}}{|\xi_{12}|^2-|\xi_{21}|^2}\right| >   |\xi_{11}|^2$, expanding the binomial in (\ref{eq:MGF_temp}) yields
%
\begin{equation}
\begin{split}
\mathcal{M}_{\gamma_{\rm IQI}}\left(s\right)=\frac{|\xi_{11}|^2}{|\xi_{11}|^2-|\xi_{22}|^2}+ & \sum_{k=0}^{\infty }\frac{\left(-1\right)^k s \, e^{\frac{|\mu_{R}|^{2}+|\nu_{R}|^{2}}{|\xi_{12}|^2 \overline{\gamma}}+s \frac{|\xi_{11}|^2}{|\xi_{12}|^2}}\left( |\xi_{12}|^2-|\xi_{21}|^2\right )^{k}}{d^{k+1} } \\  & \times \int_{0}^{|\xi_{11}|^2}{ z^{k+1} \, e^{-\frac{|\xi_{11}|^2\left( |\mu_{R}|^{2}+|\nu_{R}|^{2}\right )}{|\xi_{12}|^2 \overline{\gamma}z}-s \frac{z}{|\xi_{12}|^2}}}{\rm d}z  
\end{split}
\end{equation}
\noindent Finally, equation (\ref{eq:MGF_MC_joint_gt}) is obtained by taking $y=s {z}/{|\xi_{12}|^2}$.

\balance
\bibliographystyle{IEEEtran}
\bibliography{IEEEabrv,References}

\end{document}